\documentclass[10pt,journal,comsoc]{IEEEtran}

\usepackage{amsmath,amssymb,amsthm}
\usepackage{graphicx}
\usepackage{subfigure}
\usepackage{curves}
\usepackage{amsfonts}
\usepackage{epsfig}
\usepackage{stfloats}
\usepackage{cite}
\usepackage{algorithm}
\usepackage{algorithmic}
\usepackage{epstopdf}
\usepackage{multirow}
\usepackage{array}
\usepackage{rotating}
\usepackage{comment}

\begin{document}

\title{Polar Coding for the Cognitive Interference Channel with Confidential Messages}
\author{Mengfan~Zheng, Wen~Chen and Cong~Ling
\thanks{M. Zheng and W. Chen are with the Department of Electronic Engineering at Shanghai Jiao Tong University, Shanghai, China. Emails: \{zhengmengfan, wenchen\}@sjtu.edu.cn. C. Ling is with the Department of Electrical and Electronic Engineering at Imperial College London, United Kingdom. Email: c.ling@imperial.ac.uk.

}
}

\maketitle

\begin{abstract}
  In this paper, we propose a low-complexity, secrecy capacity achieving polar coding scheme for the cognitive interference channel with confidential messages (CICC) under the strong secrecy criterion. Existing polar coding schemes for interference channels rely on the use of polar codes for the multiple access channel, the code construction problem of which can be complicated. We show that the whole secrecy capacity region of the CICC can be achieved by simple point-to-point polar codes due to the cognitivity, and our proposed scheme requires the minimum rate of randomness at the encoder. %We further show that our proposed scheme generalizes several other multi-user polar coding problems.
  
\end{abstract}


\section{Introduction}

  Cognitive radio \cite{mitola2000cognitive} has received increasing attention due to its capability of exploiting the under-utilized spectrum resource, as the scale of wireless networks has been growing drastically nowadays. The cognitive interference channel (CIC) is a typical model for the study of cognitive radios. In this model, a primary user (can be thought of as a licensed user of a frequency band) and a cognitive user (can be thought of as an unlicensed user wishing to share the same frequency band) who has non-causal knowledge of the primary user's message transmit their own messages to their own destinations at the same time. The communication problem in the CIC has been studied in \cite{Devroye2006CR,Devroye2006CRC,wu2007CR,Jovicic2009CR,Jiang2008CR,Rini2011cognitive}. The security issue of the CIC was first considered in \cite{liang2009cognitive}, which gave the capacity-equivocation region of the CIC with confidential messages (CICC) under the weak secrecy criterion. A more general expression for the achievable rate region of the CICC with additional randomness constraints was derived in \cite{Watanabe2014cognitive} under the strong secrecy criterion, which coincides with the result of \cite{liang2009cognitive}.

  In this paper, we aim to design a polar coding scheme to achieve the whole achievable rate region of \cite{Watanabe2014cognitive}. Polar codes \cite{arikan2009channel}, originally targeted for achieving the symmetric capacity of point-to-point channels, have recently been shown to work for multi-user channels as well. It is shown that polar codes achieve the capacity regions or the known achievable rate regions of multiple access channels (MAC) \cite{arikan2012sw,sasoglu2013mac,abbe2012mmac,mahdavifar2016uniform}, broadcast channels \cite{goela2015broadcast,mondelli2015marton}, and interference channels (IC) \cite{wang2015channel,zheng2016polarIC}. In the area of physical layer security, polar codes have been shown to achieve the secrecy capacity of wiretap channels \cite{mahdavifar2011achieving,koyluoglu2012polar,sasoglu2013strong,chou2016broadcast,wei2016generalwt,gulcu2017wiretap}, and the secrecy capacity regions or the known secrecy rate regions of MAC wiretap channels \cite{chou2016MACW,wei2016generalwt}, broadcast channels with confidential messages \cite{chou2016broadcast,wei2016generalwt,gulcu2017wiretap}, IC with confidential messages \cite{wei2016generalwt}, two-way wiretap channels \cite{zheng2016tw}, and bidirectional broadcast channels with common and confidential messages \cite{andersson2013broad}. A capacity achieving secrecy polar coding scheme usually requires some eavesdropper channel information at the transmitter, which can be a drawback in practice. However, this is not a problem in the CICC since the assumption that the cognitive transmitter knows the channel information of both receivers is quite reasonable. Thus, the CICC can be a scenario where secrecy polar coding can be practically used.
  
  The CICC differs from the IC with confidential messages considered by \cite{wei2016generalwt} in that only the cognitive transmitter has confidential messages, and the cognitive transmitter has non-causal knowledge about the primary user's messages. The cognitivity not only enlarges the achievable rate region of an IC, but also can help simplify the code design.  As shown in \cite{wang2015channel,zheng2016polarIC,wei2016generalwt}, existing polar code designs for ICs that can achieve optimal rate regions require the use of the permutation based MAC polarization \cite{arikan2012sw,mahdavifar2016uniform}, as it is currently the only method that can achieve the whole achievable rate region of a MAC directly without time sharing or rate splitting. As far as we know, the practicality of this method remains open since the induced random variable by the permutation can be very complicated. In this paper, we show that the whole achievable rate region of the CICC can be achieved by point-to-point polar codes in conjunction with properly designed chaining schemes. 
  
  We summarize the contributions of this paper as follows.
  \begin{itemize}
  	\item We propose a low-complexity polar coding scheme for the general CICC that achieves the whole secrecy capacity region under the strong secrecy criterion, without any assumption on channel symmetry or degradation. 
  	\item We avoid using MAC polarization, which is a common ingredient of polar code designs for ICs but may increase system complexity, and develop a capacity achieving scheme that only requires point-to-point polar codes.
  	\item Secrecy coding schemes require a large amount of randomness in order to protect the confidential message or perform channel prefixing. We consider the randomness required by the encoder as a limited resource and show that our proposed scheme requires the minimum generating rate of randomness. 
  	\item We show that our proposed scheme is a general solution for several other multi-user polar coding problems, including the CIC without secrecy requirement.
  \end{itemize}
  
  The rest of this paper is organized as follows. In Section \ref{S:ProblemState}, we introduce the CICC model and the achievable rate region. In Section \ref{S:PolarPri}, we review some background knowledge on polar codes. Details of our proposed scheme are presented in Section \ref{S:PolarScheme}. We analyze the performance of our proposed scheme in Section \ref{S:Perf}. In Section \ref{S:Conclusion} we discuss some extensions of our proposed scheme.

  \textit{Notations:} In this paper, $[N]$ denotes the index set of $\{1,2,...,N\}$. For any $\mathcal{A}\subset [N]$, $X^{\mathcal{A}}$ denotes the subvector $\{X^i:i\in\mathcal{A}\}$ of $X^{1:N}\triangleq\{X^1,X^2,...,X^N\}$. The generator matrix of polar codes is defined as $\mathbf{G}_N=\mathbf{B}_N \textbf{F}^{\otimes n}$ \cite{arikan2009channel}, where $N=2^n$ with $n$ being an arbitrary integer, $\mathbf{B}_N$ is the bit-reversal matrix, and $\textbf{F}=
  \begin{bmatrix}
  1 & 0 \\
  1 & 1
  \end{bmatrix}$. $H_{q}(X)$ stands for the entropy of $X$ with $q$-based logarithm. 
  
\section{Problem Statement}
\label{S:ProblemState}
 \subsection{Channel Model}

\newtheorem{definition}{Definition}
 \begin{definition}
 	\label{def:CIC}
 	A 2-user CIC consists of two input alphabets $\mathcal{X}_1$ and $\mathcal{X}_2$, two output alphabets $\mathcal{Y}_1$ and $\mathcal{Y}_2$, and a probability transition function $P_{Y_1Y_2|X_1X_2}(y_1,y_2|x_1,x_2)$, where $x_1\in\mathcal{X}_1$ and $x_2\in\mathcal{X}_2$ are channel inputs of transmitter 1 and 2, respectively, and $y_1\in\mathcal{Y}_1$ and $y_2\in\mathcal{Y}_2$ are channel outputs of receiver 1 and 2, respectively. The conditional joint probability distribution of the 2-user CIC over $N$ channel uses can be factored as
 	\begin{align}
 	&P_{Y_1^{1:N}Y_2^{1:N}|X_1^{1:N}X_2^{1:N}}(y_1^{1:N},y_2^{1:N}|x_1^{1:N},x_2^{1:N})\nonumber\\
 	&~~=\prod_{j=1}^{N}P_{Y_1Y_2|X_1X_2}(y_1^j,y_2^j|x_1^j,x_2^j).
 	\end{align}
 \end{definition}

  \begin{figure}[tb]
  	\centering
  	\includegraphics[width=8cm]{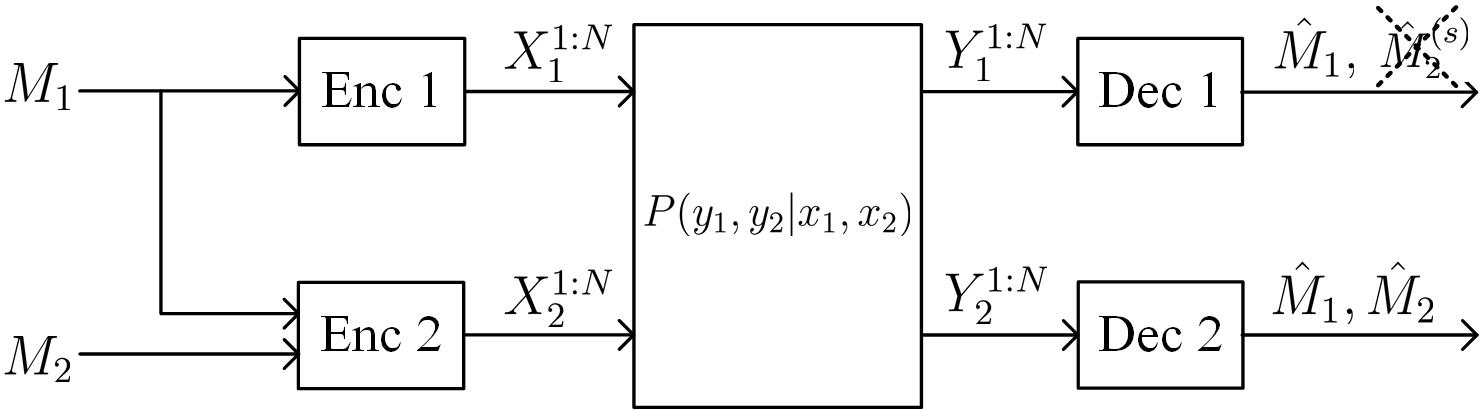}
  	\caption{The cognitive interference channel with confidential messages.} \label{fig:CICC}
  \end{figure}
   In this channel, transmitter $i$ ($i=1,2$) sends message $M_i$ to receiver $i$. Receiver 1 is required to decode $M_1$ only while receiver 2 is required to decode both $M_1$ and $M_2$\footnote{This is a case where the capacity and capacity-equivocation regions of a CIC is known \cite{liang2009cognitive}. In the general case, the capacity region of a CIC is still unknown \cite{Rini2011cognitive}.}. Since transmitter 2 has non-causal knowledge about transmitter 1's message, $M_1$ can be jointly transmitted by the two transmitters. If transmitter 2 wishes to keep part of its message (denoted as $M_2^{(s)}$) secret from receiver 1, then this model is called the CICC, as shown in Fig. \ref{fig:CICC}.
 
 \begin{definition}
 	A $(2^{NR_1},2^{NR_2},N)$ code for the 2-user CICC consists of two message sets $\mathcal{M}_1=\{1,2,...,[2^{NR_1}]\}$ and $\mathcal{M}_2=\{1,2,...,[2^{NR_2}]\}$, two encoding functions
 	\begin{equation}
 	x_1^N(m_1):\mathcal{M}_1\mapsto \mathcal{X}_1^N \text{  and  } x_2^N(m_1,m_2):\mathcal{M}_1\times\mathcal{M}_2\mapsto \mathcal{X}_2^N,
 	\end{equation}
 	and two decoding functions
 	\begin{equation}
 	\hat{m}_1(\mathbf{y}_1^N):\mathcal{Y}_1^N\mapsto \mathcal{M}_1 \text{  and  } \hat{m}_2(\mathbf{y}_2^N):\mathcal{Y}_2^N\mapsto \mathcal{M}_1\times\mathcal{M}_2.
 	\end{equation}
 \end{definition}
 
    For a given $(2^{NR_1},2^{NR_2},N)$ code for the 2-user CICC, reliability is measured by the average probability of error $P_e(N)$, defined as 
 	\begin{align}
 	&P_e(N)=\frac{1}{2^{N(R_1+R_2)}}\sum_{(M_1,M_2)\in\mathcal{M}_1\times \mathcal{M}_2} \nonumber\\
	&~~~~\mathrm{Pr}\Big{\{}\big{(}\hat{m}_1(Y_1^{1:N}),\hat{m}_2(Y_2^{1:N})\big{)}\neq (M_1,M_1,M_2)|(M_1,M_2)~\text{sent}\Big{\}},
 	\end{align}
 	where $(M_1,M_2)$ is assumed to be uniformly distributed over $\mathcal{M}_1\times \mathcal{M}_2$. Secrecy is measured the information leakage (strong secrecy)
 	\begin{equation}
 	%\label{SSC}
 	L(N)=I(Y_1^{1:N};M_2^{(s)}),
 	\end{equation}
 	or the information leakage rate (weak secrecy)
 	\begin{equation}
 	L_R(N)=\frac{1}{N}L(N).
 	\end{equation}
  
 \subsection{Achievable Rate Region}
 Let $R_{2s}$ be the transmission rate of confidential message $M_2^{(s)}$. A rate triple $(R_1,R_2,R_{2s})$ is said to be achievable for the 2-user CICC if there exists a coding scheme such that
 \begin{equation}
 \label{RC}
 \lim\limits_{N \to \infty}{P_e(N)}=0;
 \end{equation}
 and
 \begin{equation}
 \label{SSC}
 \lim\limits_{N \to \infty}{L(N)}=0 \text{ (strong secrecy), or}
 \end{equation}
 \begin{equation}
 \label{WSC}
 \lim\limits_{N \to \infty}{L_R(N)}=0 \text{ (weak secrecy)}.
 \end{equation}
 The capacity-equivocation region of the CICC (under the constraint that the cognitive receiver must decode both transmitters' messages) was derived in \cite{liang2009cognitive} and is shown below.
 \newtheorem{theorem}{Theorem}
 \begin{theorem}[\cite{liang2009cognitive}]
 \label{Capacity-1}
 The capacity-equivocation region of the CICC under the weak secrecy criterion is
 \begin{align}
 &\mathcal{R}= \bigcup_{P\in\mathcal{P}}\nonumber\\
 &\left\lbrace
 \begin{matrix}
 \left(
 \begin{array}{ccc}
 R_1\\
 R_2\\
 R_{2s}
 \end{array}
 \right)&
 \left|
 \begin{array}{ccc}
 \begin{array}{lll}
 0\leq R_1 \leq \min\{I(U;X_1;Y_1),I(U,X_1;Y_2)\}\\
 0\leq R_2 \leq I(U,V;Y_2|X_1)\\
 R_1+R_2 \leq \min\{I(U,X_1;Y_1),I(U,X_1;Y_2)\}\\
 ~~~~~~+I(V;Y_2|U,X_1)\\
 0\leq R_{2s} \leq I(V;Y_2|U,X_1)-I(V;Y_1|U,X_1)
 \end{array}
 \end{array}\right.
 \end{matrix}
 \right\rbrace,\nonumber
 \end{align}
 where $\mathcal{P}=\{P_{UVX_1X_2} \text{ factorizing as: } P_{U,V,X_1}P_{X_1|V}\}$, and the cardinality bounds for auxiliary random variables $U$ and $V$ are 
 \begin{align*}
 |\mathcal{U}|&\leq |\mathcal{X}_1|\cdot |\mathcal{X}_2|+3,\\
 |\mathcal{V}|&\leq |\mathcal{X}_1|^2\cdot |\mathcal{X}_2|^2+4|\mathcal{X}_1|\cdot |\mathcal{X}_2|+3.
 \end{align*}
\end{theorem}
  
  As we know, secrecy coding schemes require a large amount of randomness in the encoder. Reference \cite{Watanabe2014cognitive} considered the generating rate of randomness needed by the stochastic encoder as a constraint and developed a more general achievable rate region under the strong secrecy criterion as shown in Theorem \ref{Capacity-2}. In \cite{Watanabe2014cognitive}, besides the common message and the confidential message, transmitter 2's message was further divided into a private message, which should be decoded by receiver 2 but not necessarily be secret from receiver 1. 
  \begin{theorem}[\cite{Watanabe2014cognitive}]
  \label{Capacity-2}
  	
  Let $\mathcal{R}^*$ be closed convex set consisting of rate quadruples $(R_r,R_1,R_{2p},R_{2s})$ for which there exist auxiliary random variables $(U, V)$ such that
  \begin{align}
  (U,X_1)&\leftrightarrow V \leftrightarrow X_2,\label{constraint-1}\\
  (U,V)&\leftrightarrow (X_1,X_2) \leftrightarrow (Y_1,Y_2),\label{constraint-2}
  \end{align}  
  and
  \begin{align}
  R_1 &\leq \min\{I(U,X_1;Y_1),I(U,X_1;Y_2)\},\label{ineq-1}\\
  R_{2p}+R_{2s} &\leq I(U,V;Y_2|X_1),\label{ineq-2}\\
  R_1+R_{2p}+R_{2s} &\leq I(V;Y_2|U,X_1)\nonumber\\
  &~~~~+\min\{I(U,X_1;Y_1),I(U,X_1;Y_2)\},\label{ineq-3}\\
  R_{2s} &\leq I(V;Y_2|U,X_1)-I(V;Y_1|U,X_1),\label{ineq-4}\\
  R_{2p}+R_r &\geq I(X_2;Y_1|U,X_1),\label{ineq-5}\\
  R_r &\geq I(X_2;Y_1|U,V,X_1).\label{ineq-6}
  \end{align}
  Then $\mathcal{R}^*$ is an achievable rate region for the CICC. The cardinality bounds for auxiliary random variables $U$ and $V$ are the same as in Theorem \ref{Capacity-1}.
  \end{theorem}

\section{Polar Coding Preliminaries}
\label{S:PolarPri}

  Polar codes were originally invented to achieve the symmetric capacity of discrete memoryless channels (DMC) \cite{arikan2009channel}. A method for constructing polar codes for asymmetric channels without alphabet extension was introduced in \cite{honda2013asymmetric}. An improved low-complexity method was independently developed in \cite{chou2016broadcast} and \cite{gad2016asymmetric}, which overcomes the major drawback of the scheme in \cite{honda2013asymmetric} that the encoder and decoder need to share a large amount of random mappings. Now we briefly review this method.
  
  First, we introduce the lossless polar source coding techniques for arbitrary discrete memoryless sources in \cite{arikan2010source,sasoglu2011polar}. Let $(X,Y)\sim p_{X,Y}$ be a random variable pair over $(\mathcal{X}\times \mathcal{Y})$, where $X$ is the source to be compressed, $Y$ is \textit{side information} of $X$, $|\mathcal{X}|=q_X$ is a prime number\footnote{We only consider the prime number case because polarization in this case is similar to the binary case. For composite $q_X$, one needs to use some special type of quasigroup operation instead of group operation to make polarization happen \cite{sasoglu2011polar}.}, and $\mathcal{Y}$ is an arbitrary countable set. Let $U^{1:N}=X^{1:N}\mathbf{G}_N$ be a transformation on $N$ successive samples of $X$. As $N$ goes to infinity, polarization happens in the sense that $U^j$ ($j\in [N]$) becomes either almost independent of $(Y^{1:N},U^{1:j-1})$ and uniformly distributed, or almost determined by $(Y^{1:N},U^{1:j-1})$. For $\delta_N=2^{-N^\beta}$ with $\beta \in (0,1/2)$, we can define the following sets of polarized indices:
  \begin{align}
  \mathcal{H}^{(N)}_{X|Y}&=\{j\in [N]:H_{q_X}(U^j|Y^{1:N},U^{1:j-1})\geq 1-\delta_N\},\label{HXY}\\
  \mathcal{L}^{(N)}_{X|Y}&=\{j\in [N]:H_{q_X}(U^j|Y^{1:N},U^{1:j-1})\leq \delta_N\},\label{LXY}
  \end{align}
  which satisfy \cite{sasoglu2011polar,chou2016broadcast} 
  \begin{align}
  \lim_{N\rightarrow \infty}\frac{1}{N}|\mathcal{H}^{(N)}_{X|Y}|&=H_{q_X}(X|Y),\\
  \lim_{N\rightarrow \infty}\frac{1}{N}|\mathcal{L}^{(N)}_{X|Y}|&=1-H_{q_X}(X|Y).
  \end{align}

  The polarization of a single source $X\in\mathcal{X}$ can be seen as a special case of the above case by letting $\mathcal{Y}=\emptyset$. Similarly we can define the following sets of polarized indices:
  \begin{align}
  \mathcal{H}^{(N)}_X&=\{j\in [N]:H_{q_X}(U^j|U^{1:j-1})\geq 1-\delta_N\},\label{HX}\\
  \mathcal{L}^{(N)}_X&=\{j\in [N]:H_{q_X}(U^j|U^{1:j-1})\leq \delta_N\},\label{LX}
  \end{align}
  with
  \begin{align}
  \lim_{N\rightarrow \infty}\frac{1}{N}|\mathcal{H}^{(N)}_X|&=H_{q_X}(X),\\
  \lim_{N\rightarrow \infty}\frac{1}{N}|\mathcal{L}^{(N)}_X|&=1-H_{q_X}(X).
  \end{align}

  Next, we can consider polar coding for arbitrary discrete memoryless channels with the aforementioned knowledge. Let $W(Y|X)$ be a DMC with a $q_X$-ary input alphabet $\mathcal{X}$, where $q_X$ is a prime number, and an arbitrary countable output alphabet $\mathcal{Y}$. Let $U^{1:N}=X^{1:N}\mathbf{G}_N$ and define $\mathcal{H}^{(N)}_{X|Y}$, $\mathcal{L}^{(N)}_{X|Y}$, $\mathcal{H}^{(N)}_X$ and $\mathcal{L}^{(N)}_X$, as in (\ref{HXY}), (\ref{LXY}), (\ref{HX}) and (\ref{LX}), respectively. Define the \textit{information set} (or \textit{reliable set}), \textit{frozen set} and \textit{almost deterministic set} respectively as follows:
  \begin{align}
  \mathcal{I}&\triangleq \mathcal{H}_X^{(N)}\cap \mathcal{L}_{X|Y}^{(N)}, \label{PCAC-I}\\
  \mathcal{F}&\triangleq \mathcal{H}_X^{(N)}\cap (\mathcal{L}_{X|Y}^{(N)})^C, \label{PCAC-Fr}\\
  \mathcal{D}&\triangleq (\mathcal{H}_X^{(N)})^C.\label{PCAC-Fd}
  \end{align}
  
  $\mathcal{D}$ is called almost deterministic because part of its indices, $(\mathcal{H}_X^{(N)})^C\cap (\mathcal{L}_{X}^{(N)})^C$, are not fully polarized. The encoding procedure goes as follows: $\{u^j\}_{j\in \mathcal{I}}$ carry information, $\{u^j\}_{j\in \mathcal{F}}$ are filled with uniformly distributed frozen symbols (shared between the transmitter and the receiver), and $\{u^j\}_{j\in \mathcal{D}}$ are assigned by random mappings $\lambda_j(u^{1:j-1})$ which randomly generate an output $u\in\mathcal{X}$ with probability $P_{U^j|U^{1:j-1}}(u|u^{1:j-1})$. In order for the receiver to decode successfully, \cite{chou2016broadcast} and \cite{gad2016asymmetric} proposed to send part of the almost deterministic symbols, $\{u^j\}_{j\in (\mathcal{H}_X^{(N)})^C\cap (\mathcal{L}_{X|Y}^{(N)})^C}$, to the receiver with some reliable error-correcting code separately, the rate of which is shown to vanish as $N$ goes to infinity.
  
  Having received $y^{1:N}$ and recovered $\{u^j\}_{j\in (\mathcal{H}_X^{(N)})^C\cap (\mathcal{L}_{X|Y}^{(N)})^C}$, the receiver decodes $u^{1:N}$ with a successive cancellation decoder (SCD):
  \begin{align}
  &\bar{u}^{j}=\nonumber\\
  &\begin{cases}
  u^j,~~~~~~~~~~~~~~~~~~~~~~~~~~~~~~~~~~~~~~~~~~~~~~\text{if } j\in (\mathcal{L}_{X|Y}^{(N)})^C,\\
  \arg\max_{u\in\{0,1\}}P_{U^{j}|Y^{1:N}U^{1:j-1}}(u|y^{1:N},u^{1:j-1}),~~\text{if } j\in \mathcal{L}_{X|Y}^{(N)}.
  \end{cases}\nonumber
  \end{align}
  
  The transmission rate of this scheme, $R=|\mathcal{I}|/N$, is shown to achieve channel capacity \cite{honda2013asymmetric}
  \begin{equation}
  \label{PolarRate}
  \lim_{N\rightarrow \infty}R=I(X;Y).
  \end{equation}

  \section{Proposed Polar Coding Scheme}
  \label{S:PolarScheme}

  %We note that in \cite{liang2009cognitive}, transmitter 2's message $M_2$ is split into a common part intended for both receivers, and a private part intended only for receiver 2 and partly secured from receiver 1, while in \cite{Watanabe2014cognitive}, $M_2$ is split into a common part, a private part that needs to be decoded by receiver 2 but not necessarily by receiver 1, and a confidential part that must be secured from receiver 1. From Theorem \ref{Capacity-1} and \ref{Capacity-2} (also shown in \cite[Corollary 1]{Watanabe2014cognitive}) we know that these two treatments result in the same achievable rate region when there is no randomness constraint. In this paper, we take the more general treatment of \cite{Watanabe2014cognitive} to show the rate of each component message. 
  
  Our encoding scheme is illustrated in Fig. \ref{fig:encoding}. Transmitter 1's message $M_1$ is split into two parts, $M_1^{(1)}$ and $M_1^{(2)}$, carried by transmitter 1's and transmitter 2's signals respectively. Transmitter 2's message $M_2$ is split into three parts, a common message $M_2^{(c)}$ intended for both receivers, a private message $M_2^{(p)}$ intended only for receiver 2, and a confidential message $M_2^{(s)}$ intended only for receiver 2 and must be secured from receiver 1. Details of transmitter 2's encoding procedure are as follows. $M_1^{(2)}$ and $M_2^{(c)}$ are encoded into $U^{1:N}$ first, $M_2^{(p)}$ and $M_2^{(s)}$ are then superimposed on $(U^{1:N},X_1^{1:N})$ and encoded into $V^{1:N}$ (known as superposition coding). Finally, randomness $M_R$ is added to $V^{1:N}$ to generate transmitter 2's final codeword $X_2^{1:N}$ (known as channel prefixing). Note that $X_1^{1:N}$ can be seen as the known interference to transmitter 2. Thus, this superposition coding scheme also involves the idea of dirty paper coding. In the rest of this section, the rates of $M_1^{(1)}$, $M_1^{(2)}$, $M_2^{(c)}$, $M_2^{(p)}$ and $M_2^{(s)}$ will be denoted by $R_1^{(1)}$, $R_1^{(2)}$, $R_2^{(c)}$, $R_2^{(p)}$ and $R_2^{(s)}$, respectively. Notice that in Theorem \ref{Capacity-2}, $R_{2p}$ is the sum of $R_2^{(c)}$ and $R_2^{(p)}$ defined here.
  \begin{figure}[tb]
  	\centering
  	\includegraphics[width=8.5cm]{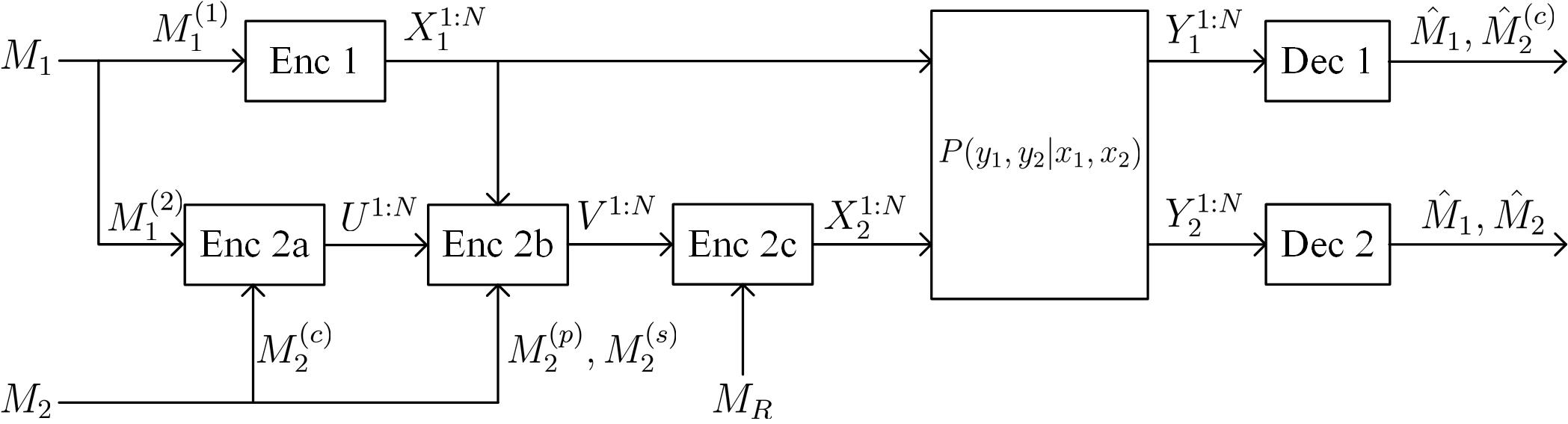}
  	\caption{Encoding scheme for the CICC.} \label{fig:encoding}
  \end{figure}
  
  A common problem of polar code designs for general multi-user channels is that the polarized sets for different receivers are usually not aligned (meaning that there is no inclusion relation among them), making it difficult to achieve the optimal rate region within a single transmission block. In this paper, we adopt the chaining method \cite{hassani2014universal} which connects a series of encoding blocks to solve this problem. In our scheme, $m$ ($m\geq 1$) encoding blocks are chained into a \textit{frame}, and two receivers decode a frame in reverse orders. Our scheme is designed in such a way that receiver 1 (the primary receiver) decodes in the natural order (i.e., from block 1 to block $m$) while receiver 2 (the secondary receiver) decodes in the reverse order (i.e., from block $m$ to block 1).
  
  In this paper, we only discuss the case when random variables $X_1$, $X_2$, $U$ and $V$ all have prime alphabets. Suppose $q_{X_1}=|\mathcal{X}_1|$ and $q_{X_2}=|\mathcal{X}_2|$ are two prime numbers, $q_{U}=|\mathcal{U}|$ is the smallest prime number larger than $|\mathcal{X}_1|\cdot |\mathcal{X}_2|+3$, and $q_{V}=|\mathcal{V}|$ is the smallest prime number larger than $|\mathcal{X}_1|^2\cdot |\mathcal{X}_2|^2+4|\mathcal{X}_1|\cdot |\mathcal{X}_2|+3$. Consider a random variable tuple $(U,V,X_1,X_2,Y_1,Y_2)$ with joint distribution $P_{UVX_1X_2Y_1Y_2}$ that satisfy (\ref{constraint-1}) and (\ref{constraint-2}). The goal of our proposed scheme is to achieve every equation in (\ref{ineq-1})--(\ref{ineq-6}).
  
  \subsection{Common Message Encoding}
  \label{S:CME}
  
  Common messages $M_1^{(1)}$, $M_1^{(2)}$ and $M_2^{(c)}$ are encoded into two sequences of random variables, $X_1^{1:N}$ and $U^{1:N}$. At first glance, we may synthesize two MACs, $P(Y_1|X_1,U)$ and $P(Y_2|X_1,U)$, and design a polar code that works for both of them. This approach requires the use of MAC polarization or rate splitting techniques which will increase the complexity of the scheme. Note that there is no major difference between $M_1^{(2)}$ and $M_2^{(c)}$ in regard to our encoding scheme. They only affects the rate allocation between $M_1$ and $M_2$. Due to this flexibility, we will show that the whole region can be achieved with simple point-to-point polar codes. For simplicity, define $R^{(c)}=R_1^{(2)}+R_2^{(c)}$. 
  From (\ref{ineq-1}) we have
  \begin{equation}
  \label{R-C}
  R_1^{(1)}+R^{(c)}\leq\min\{I(U,X_1;Y_1),I(U,X_1;Y_2)\}.
  \end{equation}
  
  Our approach consists of two layers of coding. In the first layer, $M_1^{(1)}$ is encoded into $X_1^{1:N}$, treating random variable $U$ as noise. In the second layer, $M_1^{(2)}$ and $M_2^{(c)}$ are encoded into $U^{1:N}$ with $X_1^{1:N}$ being treated as side information. The receivers decode $X_1^{1:N}$ first, and then utilize the estimation of $X_1^{1:N}$ to decode $U^{1:N}$. In a conventional MAC case, such an approach can only achieve a corner point of the achievable rate region of a MAC. However, under the cognitive setting, since transmitter 2 knows transmitter 1's message non-causally, and the message carried by $U^{1:N}$ can be allocated flexibly between two transmitters, this approach can achieve more points. We first show how to achieve $R_1^{(1)}+R^{(c)}=\min\{I(U,X_1;Y_1),I(U,X_1;Y_2)\}$ in this subsection and $R_2^{(p)}+R_2^{(s)}=I(U,V;Y_2|X_1)$ in the next subsection, and then prove in Section \ref{S:ARR} that other rate pairs in the achievable rate region in Theorem \ref{Capacity-2} can be achieved by adjusting the ratio between $M_1^{(2)}$ and $M_2^{(c)}$.
  
  Let  $U_1^{1:N}=X_1^{1:N}\mathbf{G}_N$ and $U^{'1:N}=U^{1:N}\mathbf{G}_N$. For $\delta_N=2^{-N^\beta}$ with $\beta \in (0,1/2)$, define the following polarized sets:
  \begin{equation}
  \begin{aligned}
    \mathcal{H}^{(N)}_{X_1}&\triangleq \big{\{}j\in [N]:H_{q_{X_1}}(U_1^j|U_1^{1:j-1})\geq 1-\delta_N \big{\}},\\
  \mathcal{L}^{(N)}_{X_1|Y_1}&\triangleq \big{\{}j\in [N]:H_{q_{X_1}}(U_1^j|Y_1^{1:N}, U_1^{1:j-1})\leq \delta_N \big{\}},\\
  \mathcal{L}^{(N)}_{X_1|Y_2}&\triangleq \big{\{}j\in [N]:H_{q_{X_1}}(U_1^j|Y_2^{1:N}, U_1^{1:j-1})\leq \delta_N \big{\}},\\
  \mathcal{H}^{(N)}_{U|X_1}&\triangleq \big{\{}j\in [N]:H_{q_U}(U^{'j}|X_1^{1:N},U^{'1:j-1})\geq 1-\delta_N \big{\}},\\
  \mathcal{L}^{(N)}_{U|Y_1X_1}&\triangleq \big{\{}j\in [N]:H_{q_U}(U^{'j}|Y_1^{1:N},X_1^{1:N}, U^{'1:j-1})\leq \delta_N \big{\}},\\
  \mathcal{L}^{(N)}_{U|Y_2X_1}&\triangleq \big{\{}j\in [N]:H_{q_U}(U^{'j}|Y_2^{1:N},X_1^{1:N}, U^{'1:j-1})\leq \delta_N \big{\}}.
  \end{aligned}
  \end{equation}
  Then define the following sets of indices for $U_1^{1:N}$:
  \begin{equation}
  \begin{aligned}
  \label{ConstructCommon1}
  \mathcal{I}_{1c}^{(1)}&=\mathcal{H}^{(N)}_{X_1} \cap \mathcal{L}^{(N)}_{X_1|Y_1},\\
  \mathcal{I}_{1c}^{(2)}&=\mathcal{H}^{(N)}_{X_1} \cap \mathcal{L}^{(N)}_{X_1|Y_2},\\
  \mathcal{F}_{1c}&=\mathcal{H}^{(N)}_{X_1} \cap \big{(}\mathcal{L}^{(N)}_{X_1|Y_1}\big{)}^C \cap \big{(}\mathcal{L}^{(N)}_{X_1|Y_2}\big{)}^C,\\
  \mathcal{D}_{1c}&=\big{(}\mathcal{H}^{(N)}_{X_1}\big{)}^C,
  \end{aligned}
  \end{equation}
  where $\mathcal{I}_{1c}^{(1)}$ and $\mathcal{I}_{1c}^{(2)}$ are the reliable sets for receiver 1 and 2 respectively, $\mathcal{F}_{1c}$ is the intersection of two receivers' frozen sets, and $\mathcal{D}_{1c}$ is the almost deterministic set. Similarly define
  \begin{equation}
  \begin{aligned}
  \label{ConstructCommon2}
  \mathcal{I}_{2c}^{(1)}&=\mathcal{H}^{(N)}_{U|X_1} \cap \mathcal{L}^{(N)}_{U|Y_1X_1},\\
  \mathcal{I}_{2c}^{(2)}&=\mathcal{H}^{(N)}_{U|X_1} \cap \mathcal{L}^{(N)}_{U|Y_2X_1},\\
  \mathcal{F}_{2c}&=\mathcal{H}^{(N)}_{U|X_1} \cap \big{(}\mathcal{L}^{(N)}_{U|Y_1X_1}\big{)}^C \cap \big{(}\mathcal{L}^{(N)}_{U|Y_2X_1}\big{)}^C,\\
  \mathcal{D}_{2c}&=\big{(}\mathcal{H}^{(N)}_{U|X_1}\big{)}^C.
  \end{aligned}
  \end{equation} 
  for $U^{'1:N}$. From (\ref{PolarRate}) we have
  \begin{equation}
  \begin{aligned}
  \label{CommonRate}
  \lim\limits_{N\rightarrow \infty} \frac{1}{N}|\mathcal{I}_{1c}^{(1)}|  &=I(X_1;Y_1),~~\lim\limits_{N\rightarrow \infty} \frac{1}{N}|\mathcal{I}_{2c}^{(1)}|=I(U;Y_1|X_1),   \\
  \lim\limits_{N\rightarrow \infty} \frac{1}{N}|\mathcal{I}_{1c}^{(2)}| &=I(X_1;Y_2),~~
  \lim\limits_{N\rightarrow \infty} \frac{1}{N}|\mathcal{I}_{2c}^{(2)}| =I(U;Y_2|X_1).
  \end{aligned}  
  \end{equation}
  
  If we design two separate chaining schemes for $U^{1:N}$ and $U^{'1:N}$ respectively, it is easy to verify that the achievable common message rate is
  \begin{align}
  R_1^{(1)}+R^{(c)}&\leq\min\{I(X_1;Y_1),I(X_1;Y_2)\}\nonumber\\
  &~~~~+\min\{I(U;Y_1|X_1),I(U;Y_2|X_2)\}.  \label{R-C-s}
  \end{align}
  Such a scheme achieves (\ref{R-C}) only in the following two cases:
  \begin{itemize}
  	\item (Case 1) $I(X_1;Y_1)\leq I(X_1;Y_2)$ and $I(U;Y_1|X_1)\leq I(U;Y_2|X_1)$,
  	\item (Case 2) $I(X_1;Y_1)\geq I(X_1;Y_2)$ and $I(U;Y_1|X_1)\geq I(U;Y_2|X_1)$.
  \end{itemize}
  In the other two cases of
  \begin{itemize}
  	\item (Case 3) $I(X_1;Y_1)< I(X_1;Y_2)$ and $I(U;Y_1|X_1)> I(U;Y_2|X_1)$,
  	\item (Case 4) $I(X_1;Y_1)> I(X_1;Y_2)$ and $I(U;Y_1|X_1)< I(U;Y_2|X_1)$,
  \end{itemize}
  the achievable rate in (\ref{R-C-s}) is strictly smaller than that in (\ref{R-C}). In these cases, the chaining scheme should be jointly designed for $U^{1:N}$ and $U^{'1:N}$, which we refer to as cross-transmitter chaining.

  \subsubsection{Case 1 and Case 2}
  Since Case 2 is similar to Case 1 by swapping the roles of two transmitters, we only describe the chaining scheme in Case 1 for brevity. From (\ref{CommonRate}) we know that given sufficiently large $N$, we always have $|\mathcal{I}_{1c}^{(1)}|\leq |\mathcal{I}_{1c}^{(2)}|$ and $|\mathcal{I}_{2c}^{(1)}|\leq |\mathcal{I}_{2c}^{(2)}|$. Define
  \begin{equation}
  \label{I-c}
  \begin{aligned}
  \mathcal{I}_{1c}^{(0)}=\mathcal{I}_{1c}^{(1)}\cap\mathcal{I}_{1c}^{(2)},~~\mathcal{I}_{1c}^{(1a)}=\mathcal{I}_{1c}^{(1)}\setminus\mathcal{I}_{1c}^{(2)},\\
  \mathcal{I}_{2c}^{(0)}=\mathcal{I}_{2c}^{(1)}\cap\mathcal{I}_{2c}^{(2)},~~\mathcal{I}_{2c}^{(1a)}=\mathcal{I}_{2c}^{(1)}\setminus\mathcal{I}_{2c}^{(2)}.
  \end{aligned}
  \end{equation} 
  Choose an arbitrary subset $\mathcal{I}_{1c}^{(2a)}$ of $\mathcal{I}_{1c}^{(2)}\setminus \mathcal{I}_{1c}^{(1)}$ such that $|\mathcal{I}_{1c}^{(2a)}|=|\mathcal{I}_{1c}^{(1a)}|$, and an arbitrary subset $\mathcal{I}_{2c}^{(2a)}$ of $\mathcal{I}_{2c}^{(2)}\setminus \mathcal{I}_{2c}^{(1)}$ such that $|\mathcal{I}_{2c}^{(2a)}|=|\mathcal{I}_{2c}^{(1a)}|$. The chaining scheme goes as follows.
  
  (I) In the 1st block, transmitter 1 encodes its common message as:
  \begin{itemize}
  	\item $\{u_1^{j}\}_{j\in \mathcal{I}_{1c}^{(0)}\cup\mathcal{I}_{1c}^{(1a)}}$ store message symbols from $M_1^{(1)}$,
  	\item $\{u_1^{j}\}_{j\in (\mathcal{I}_{1c}^{(0)}\cup\mathcal{I}_{1c}^{(1a)}\cup \mathcal{D}_{1c})^C}$ carry frozen symbols uniformly distributed over $\mathcal{X}_1$,
  	\item $\{u_1^{j}\}_{j\in \mathcal{D}_{1c}}$ are assigned by random mappings $\lambda_j(u_1^{1:j-1})$ that generate an output $u\in\mathcal{X}_1$ according to conditional probability $P_{U_1^{j}|U_1^{1:j-1}}(u|u_1^{1:j-1})$,
  \end{itemize}  
  and transmitter 2 encodes its common message as:
  \begin{itemize}
  	\item $\{u^{'j}\}_{j\in \mathcal{I}_{2c}^{(0)}\cup\mathcal{I}_{2c}^{(1a)}}$ store message symbols from $M_1^{(2)}$ and $M_2^{(c)}$,
  	\item $\{u^{'j}\}_{j\in (\mathcal{I}_{2c}^{(0)}\cup\mathcal{I}_{2c}^{(1a)}\cup \mathcal{D}_{2c})^C}$ carry frozen symbols uniformly distributed over $\mathcal{U}$,
  	\item $\{u^{'j}\}_{j\in \mathcal{D}_{2c}}$ are assigned by random mappings $\lambda_j(u^{'1:j-1})$ that generate an output $u\in\mathcal{U}$ according to conditional probability $P_{U^{'j}|U^{'1:j-1}}(u|u^{'1:j-1})$.
  \end{itemize}  
  
  (II) In the $i$th  ($1<i<m$) block, transmitter 1 assigns $\{u_1^{j}\}_{j\in \mathcal{I}_{1c}^{(2a)}}$ with the same value of $\{u_1^{j}\}_{j\in \mathcal{I}_{1c}^{(1a)}}$ in block $i-1$, and transmitter 2 assigns $\{u^{'j}\}_{j\in \mathcal{I}_{2c}^{(2a)}}$ with the same value of $\{u^{'j}\}_{j\in \mathcal{I}_{2c}^{(1a)}}$ in block $i-1$. The rest of $u_1^{1:N}$ and $u^{'1:N}$ are determined in the same way as in (I).
  
  (III) In the $m$th block, transmitter 1 assigns $\{u_1^{j}\}_{j\in \mathcal{I}_{1c}^{(1a)}}$ with frozen symbols uniformly distributed over $\mathcal{X}_1$, and transmitter 2 assigns $\{u^{'j}\}_{j\in \mathcal{I}_{2c}^{(1a)}}$ with frozen symbols uniformly distributed over $\mathcal{U}$. The rest of $u_1^{1:N}$ and $u^{'1:N}$ are determined in the same way as in (II).
  
  \begin{figure}[tb]
  	\centering
  	\includegraphics[width=9cm]{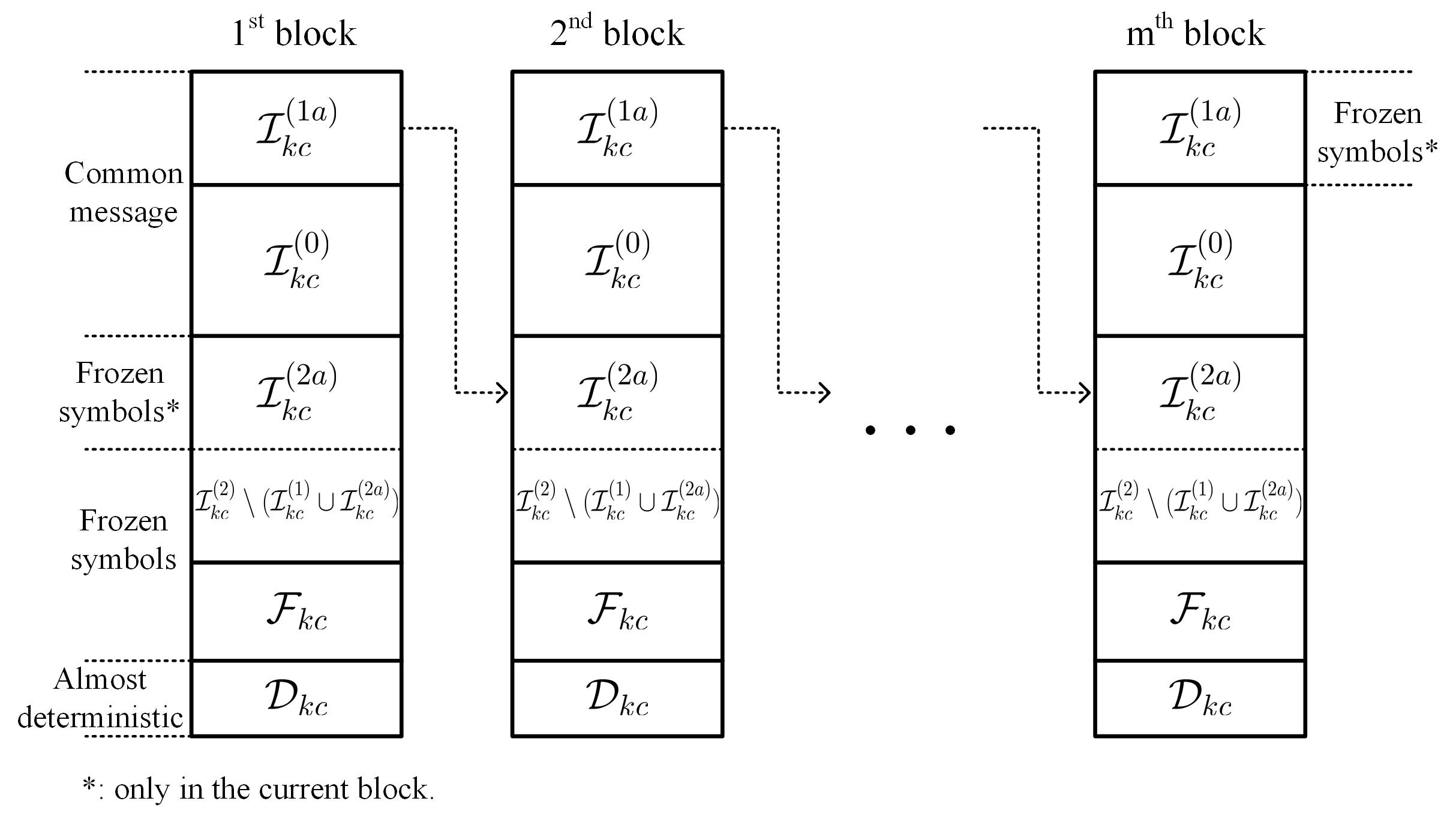}
  	\caption{The chaining scheme of transmitter $k$ ($k=1,2$) for common messages in Case 1.} \label{fig:chaining-c1}
  \end{figure}
  The chaining scheme in Case 1 is shown in Fig. \ref{fig:chaining-c1}.
  After each transmission block, transmitter 1 additionally sends a vanishing fraction of the almost deterministic symbols, $\{u_1^{j}\}_{j\in (\mathcal{H}_{X_1}^{(N)})^C\cap (\mathcal{L}_{X_1|Y_1}^{(N)})^C}$ and $\{u_1^{j}\}_{j\in (\mathcal{H}_{X_1}^{(N)})^C\cap (\mathcal{L}_{X_1|Y_2}^{(N)})^C}$, to receiver 1 and 2 respectively with some reliable error-correcting code. Similarly, transmitter 2 sends $\{u^{'j}\}_{j\in (\mathcal{H}_{U|X_1}^{(N)})^C\cap (\mathcal{L}_{U|Y_1X_1}^{(N)})^C}$ and $\{u^{'j}\}_{j\in (\mathcal{H}_{U|X_1}^{(N)})^C\cap (\mathcal{L}_{U|Y_2X_1}^{(N)})^C}$ to two receivers respectively after each block. From Section \ref{S:PolarPri} we know that the rate for transmitting these symbols vanishes as $N$ increases. Thus, the cost for these extra transmissions can be made negligible. Also note that frozen symbols at $\mathcal{F}_{kc}\cup\big{(} \mathcal{I}_{kc}^{(2)}\setminus(\mathcal{I}_{kc}^{(1)}\cup\mathcal{I}_{kc}^{(2a)})\big{)}$ ($k=1,2$) can be reused since they only need to be independently and uniformly distributed. Thus, the rate of frozen symbols which must be shared between transmitters and receivers can be made negligible as well by reusing them over sufficient number of blocks. In the analysis part in the next section we will assume that this part of frozen symbols are the same for all blocks in a frame.

  \subsubsection{Case 3 and Case 4}
  \label{S:Case34}

  Since Case 4 is similar to Case 3 by swapping the roles of two transmitters, we only describe the chaining scheme in Case 3 for brevity. In this case, given sufficiently large $N$, we always have $|\mathcal{I}_{1c}^{(1)}|\leq |\mathcal{I}_{1c}^{(2)}|$ and $|\mathcal{I}_{2c}^{(1)}|\geq |\mathcal{I}_{2c}^{(2)}|$. 
  
  If $\min\{I(U,X_1;Y_1),I(U,X_1;Y_2)\}=I(U,X_1;Y_1)$,
  which we refer to as Case 3-1, we have $|\mathcal{I}_{1c}^{(2)}|-|\mathcal{I}_{1c}^{(1)}|\geq |\mathcal{I}_{2c}^{(1)}|- |\mathcal{I}_{2c}^{(2)}|$ given sufficiently large $N$. In this case, define $\mathcal{I}_{1c}^{(0)}$, $\mathcal{I}_{2c}^{(0)}$, $\mathcal{I}_{1c}^{(1a)}$ and $\mathcal{I}_{2c}^{(1a)}$ in the same way as in (\ref{I-c}), and define
  \begin{equation}
  \mathcal{I}_{2c}^{(2a)}=\mathcal{I}_{2c}^{(2)}\setminus \mathcal{I}_{2c}^{(1)}.
  \end{equation}
  Choose an arbitrary subset $\mathcal{I}_{2c}^{(1b)}$ of $\mathcal{I}_{2c}^{(1a)}$ with $|\mathcal{I}_{2c}^{(1b)}|=|\mathcal{I}_{2c}^{(1)}|-|\mathcal{I}_{2c}^{(2)}|$, and an arbitrary subset $\mathcal{I}_{1c}^{(2a)}$ of $\mathcal{I}_{1c}^{(2)}\setminus \mathcal{I}_{1c}^{(1)}$ with $|\mathcal{I}_{1c}^{(2a)}|=|\mathcal{I}_{1c}^{(1a)}|+|\mathcal{I}_{2c}^{(1b)}|$. Let $\mathcal{I}_{1c}^{(2b)}$ be a subset of $\mathcal{I}_{1c}^{(2a)}$ with the same size as $\mathcal{I}_{2c}^{(1b)}$. The chaining scheme in Case 3-1 goes as follows and is illustrated in Fig. \ref{fig:chaining-c3}.
  \begin{figure}[tb]
  	\centering
  	\includegraphics[width=8.5cm]{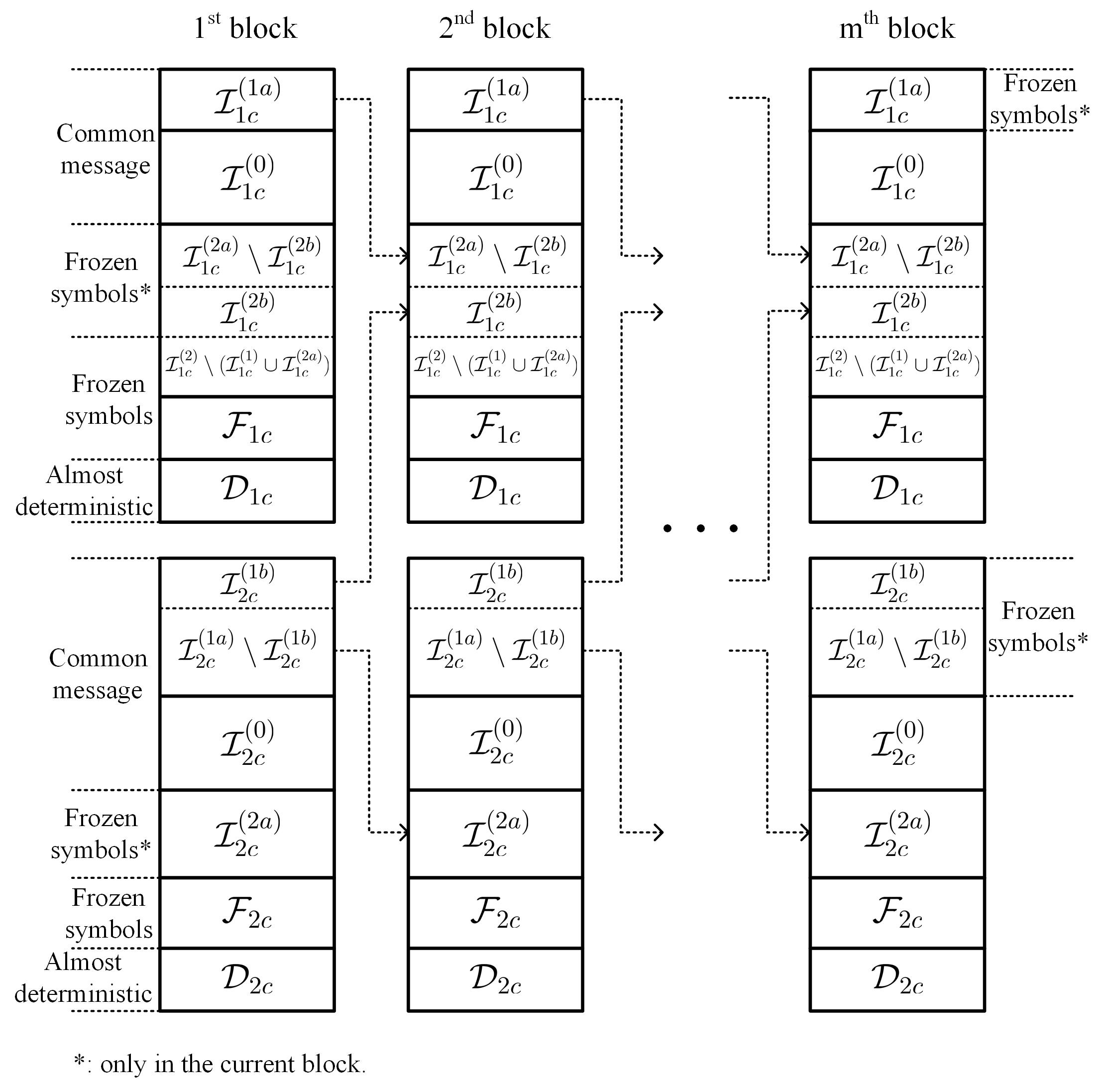}
  	\caption{The chaining scheme for common messages in Case 3-1.} \label{fig:chaining-c3}
  \end{figure}
    
  (I) In the 1st block, the encoding procedure is similar to Case 1, except that $\{u^{'j}\}_{j\in \mathcal{I}_{2c}^{(1b)}}$ of transmitter 2 are filled with message symbols from $M_1^{(1)}$ only, as they will be chained with transmitter 1's next encoding block.
  
  (II) In the $i$th  ($1<i<m$) block, transmitter 1 assigns $\{u_1^{j}\}_{j\in \mathcal{I}_{1c}^{(2a)}\setminus \mathcal{I}_{1c}^{(2b)}}$ with the same value of $\{u_1^{j}\}_{j\in \mathcal{I}_{1c}^{(1a)}}$ in block $i-1$, and $\{u_1^{j}\}_{j\in \mathcal{I}_{1c}^{(2b)}}$ with the same value of $\{u^{'j}\}_{j\in \mathcal{I}_{2c}^{(1b)}}$ in block $i-1$, while transmitter 2 assigns $\{u^{'j}\}_{j\in \mathcal{I}_{2c}^{(2a)}}$ with the same value of $\{u^{'j}\}_{j\in \mathcal{I}_{2c}^{(1a)}\setminus \mathcal{I}_{2c}^{(1b)}}$ in block $i-1$. The rest of $u_1^{1:N}$ and $u^{'1:N}$ are determined in the same way as in (I).
  
  (III) In the $m$th block, transmitter 1 assigns $\{u_1^{j}\}_{j\in \mathcal{I}_{1c}^{(1a)}}$ with frozen symbols uniformly distributed over $\mathcal{X}_1$, and transmitter 2 assigns $\{u^{'j}\}_{j\in \mathcal{I}_{2c}^{(1a)}}$ with frozen symbols uniformly distributed over $\mathcal{U}$. The rest of $u_1^{1:N}$ and $u^{'1:N}$ are determined in the same way as in (II).
  
  Note that the cross-transmitter chaining scheme described above does not violate the assumption that transmitter 1 is not cognitive, as the common message used for the cross-transmitter chaining only comes from $M_1$, of which both transmitters have non-causal knowledge.
  
  Otherwise if $\min\{I(U,X_1;Y_1),I(U,X_1;Y_2)\}=I(U,X_1;Y_2)$, 
  which we refer to as Case 3-2, we have $|\mathcal{I}_{1c}^{(2)}|-|\mathcal{I}_{1c}^{(1)}|\leq |\mathcal{I}_{2c}^{(1)}|- |\mathcal{I}_{2c}^{(2)}|$ given sufficiently large $N$. The chaining scheme in this case is similar to that in Fig. \ref{fig:chaining-c3} with the two transmitters exchanging their roles. 
  
  Similar to Case 1, two transmitters send part of their almost deterministic symbols to two receivers with some reliable error-correcting code after each block. Also, transmitter 1's frozen symbols at $\mathcal{F}_{1c}\cup\big{(} \mathcal{I}_{1c}^{(2)}\setminus(\mathcal{I}_{1c}^{(1)}\cup\mathcal{I}_{1c}^{(2a)})\big{)}$ and transmitter 2's frozen symbols at $\mathcal{F}_{2c}$ can be reused over different blocks.

  \subsection{Private and Confidential Message Encoding}
  \label{S:PCME}
  
  Since private message $M_2^{(p)}$ and confidential message $M_2^{(s)}$ are superimposed on $(U^{1:N},X_1^{1:N})$ by auxiliary random variable $V^{1:N}$, we treat $(U^{1:N},X_1^{1:N})$ as side information when applying polarization on $V^{1:N}$. Let $V^{'1:N}=V^{1:N}\mathbf{G}_N$. For $\delta_N=2^{-N^\beta}$ with $0<\beta<1/2$, define
  \begin{equation}
  \begin{aligned}
  \mathcal{H}^{(N)}_{V|X_1U}&\triangleq \big{\{}j\in[N]:H_{q_V}(V^{'j}|X_1^{1:N},U^{1:N},V^{'1:j-1})\geq 1-\delta_N \big{\}},\\
  \mathcal{H}^{(N)}_{V|Y_1X_1U}&\triangleq \big{\{}j\in[N]:H_{q_V}(V^{'j}|Y_1^{1:N},X_1^{1:N},U^{1:N},V^{'1:j-1})\\
  &~~~~~~~~~~~~~~~~~~~~~~~~~~~~~~~~~~~~~~~~~~~\geq 1-\delta_N \big{\}},\\
  \mathcal{L}^{(N)}_{V|Y_2X_1U}&\triangleq \big{\{}j\in[N]:H_{q_V}(V^{'j}|Y_2^{1:N},X_1^{1:N},U^{1:N},V^{'1:j-1})\leq \delta_N \big{\}}.
  \end{aligned}  
  \end{equation}
  Partition the indices of $V^{'1:N}$ as follows:
  \begin{equation}
  \begin{aligned}
  \label{Construct}
  \mathcal{I}_{2s}&=\mathcal{H}^{(N)}_{V|X_1U} \cap \mathcal{L}^{(N)}_{V|Y_2X_1U} \cap \mathcal{H}^{(N)}_{V|Y_1X_1U},\\
  \mathcal{I}_{2p}&=\mathcal{H}^{(N)}_{V|X_1U} \cap \mathcal{L}^{(N)}_{V|Y_2X_1U} \cap \big{(}\mathcal{H}^{(N)}_{V|Y_1X_1U}\big{)}^C,\\
  \mathcal{F}_2&=\mathcal{H}^{(N)}_{V|X_1U} \cap \big{(}\mathcal{L}^{(N)}_{V|Y_2X_1U}\big{)}^C \cap \mathcal{H}^{(N)}_{V|Y_1X_1U},\\
  \mathcal{R}_2&=\mathcal{H}^{(N)}_{V|X_1U} \cap \big{(}\mathcal{L}^{(N)}_{V|Y_2X_1U}\big{)}^C \cap \big{(}\mathcal{H}^{(N)}_{V|Y_1X_1U}\big{)}^C,\\
  \mathcal{D}_2&=\big{(}\mathcal{H}^{(N)}_{V|X_1U}\big{)}^C,
  \end{aligned}
  \end{equation}
  where $\mathcal{I}_{2s}$ is the reliable and secure set, $\mathcal{I}_{2p}$ is the reliable but insecure set, $\mathcal{R}_2$ is the unreliable and insecure set, $\mathcal{F}_2$ is the frozen set, and $\mathcal{D}_{2c}$ is the almost deterministic set.
  
  The aim of using the chaining method is to deal with the unreliable and insecure set $\mathcal{R}_2$. Consider the positive secrecy rate case (i.e., the right-hand-side of (\ref{ineq-4}) is positive). In this case, $|\mathcal{I}_{2s}|>|\mathcal{R}_2|$ always holds for sufficiently large $N$. Choose a subset $\mathcal{I}_{2s}^{(2)}$ of $\mathcal{I}_{2s}$ such that $|\mathcal{I}_{2s}^{(2)}|=|\mathcal{R}_2|$. Denote $\mathcal{I}_{2s}^{(1)}=\mathcal{I}_{2s}\setminus \mathcal{I}_{2s}^{(2)}$. The chaining scheme for transmitter 2's private and confidential messages is also designed in such a way that receiver 2 decodes from block $m$ to block 1, same as its decoding order for common messages. Details of the scheme are as follows and shown in Fig. \ref{fig:chaining-P}. 
  \begin{figure}[tb]
  	\centering
  	\includegraphics[width=9cm]{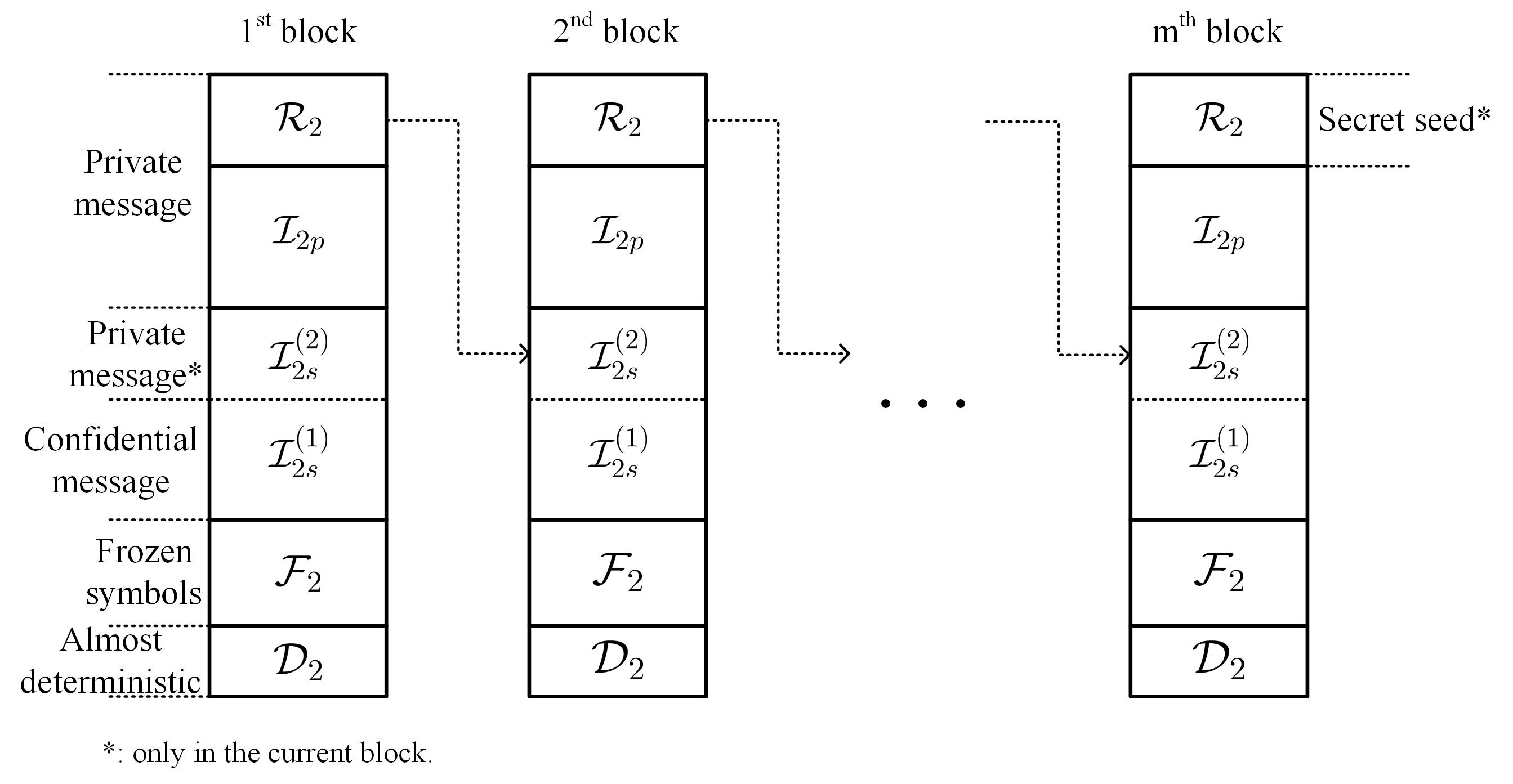}
  	\caption{The chaining scheme for transmitter 2's private and confidential messages.} \label{fig:chaining-P}
  \end{figure}
  
  (I) In the $1$st block, 
  \begin{itemize}
  	\item $\{v^{'j}\}_{j\in \mathcal{I}_{2s}^{(1)}}$ carry confidential message symbols,
  	\item $\{v^{'j}\}_{j\in \mathcal{I}_{2p}\cup\mathcal{I}_{2s}^{(2)}\cup\mathcal{R}_2}$ carry private message symbols,
  	\item $\{v^{'j}\}_{j\in \mathcal{F}_2}$ are filled with uniformly distributed frozen symbols,
  	\item $\{v^{'j}\}_{j\in \mathcal{D}_2}$ are assigned by random mappings $\lambda_j(v^{'1:j-1})$ that generate an output $v\in\mathcal{V}$ according to conditional probability $P_{V^{'j}|X_1^{1:N}U^{1:N}V^{'1:j-1}}$,
  \end{itemize}

  (II) In the $i$th $(1< i< m)$ block, $\{v^{'j}\}_{i\in \mathcal{I}_{2s}^{(2)}}$ are assigned with the same value as $\{v^{'j}\}_{i\in \mathcal{R}_2}$in the $(i-1)$th block, and the rest of $v^{'1:N}$ are determined in the same way as in (I).
  
  (III) In the $m$th block, $\{v^{'j}\}_{j\in \mathcal{R}_2}$ carry some uniformly distributed random symbols that are shared only between transmitter 2 and receiver 2 (known as secret seed), and the rest of $v^{'1:N}$ are determined in the same way as in (II).
  
  The secret seed rate can be made arbitrarily small by increasing the number of chained blocks in a frame. After each transmission block, transmitter 2 additionally sends a vanishing fraction of the almost deterministic symbols, $\{v^{'j}\}_{j\in (\mathcal{H}_{V|X_1U}^{(N)})^C\cap (\mathcal{L}_{V|Y_2X_1U}^{(N)})^C}$, to receiver 2 secretly with some reliable error-correcting code. Note that unlike in the common message encoding, the additional transmission for the almost deterministic symbols here must be kept secret from receiver 1. Nevertheless, the rate of this transmission can be made arbitrarily small by increasing $N$. Similar to the common message encoding, frozen symbols at $\mathcal{F}_{2}$ can also be reused over different blocks. In the next section we will show that with the reuse of frozen symbols, our proposed scheme still achieves strong secrecy.

  \subsection{Channel Prefixing}
  
  To generate the final codeword $X_2^{1:N}$ for transmitter 2, one can transmit $(X_1^{1:N},U^{1:N},V^{1:N})$ through a virtual channel with transition probability $P_{X_2|X_1UV}$. Also, one can consider $X_2$ and $(X_1,U,V)$ as correlated sources and apply polar source coding to obtain the final codeword. To design a scheme that requires the minimum generating rate of randomness, we take the latter approach in this paper. Let $U_2^{1:N}=X_2^{1:N}\mathbf{G}_N$. For $\delta_N=2^{-N^\beta}$ with $0<\beta<1/2$, define
  \begin{equation}
  \begin{aligned}
  \mathcal{H}^{(N)}_{X_2|X_1UV}&\triangleq \big{\{}j\in[N]:H_{q_{X_2}}(U_2^{j}|X_1^{1:N},U^{1:N},V^{1:N},U_2^{1:j-1})\\
  &~~~~~~~~~~~~~~~~~~~~~~~~~~~~~~~~~~~~~~~~~~~\geq 1-\delta_N \big{\}},\\
  \mathcal{H}^{(N)}_{X_2|Y_1X_1UV}&\triangleq \big{\{}j\in[N]:H_{q_{X_2}}(U_2^{j}|Y_1^{1:N},X_1^{1:N},U^{1:N},V^{1:N},\\
  &~~~~~~~~~~~~~~~~~~~~~~~~~~~~~~~~U_2^{1:j-1})\geq 1-\delta_N \big{\}},\\
  \mathcal{L}^{(N)}_{X_2|Y_1X_1UV}&\triangleq \big{\{}j\in[N]:H_{q_{X_2}}(U_2^{j}|Y_1^{1:N},X_1^{1:N},U^{1:N},V^{1:N},\\
  &~~~~~~~~~~~~~~~~~~~~~~~~~~~~~~~~~~~~U_2^{1:j-1})\leq \delta_N \big{\}}.
  \end{aligned}  
  \end{equation}
  Once $(X_1^{1:N},U^{1:N},V^{1:N})$ is determined, $X_2^{1:N}$ can be obtained as follows. Let $w_r$ be a random sequence uniformly distributed over $\mathcal{X}_2$ and of length $|\mathcal{H}^{(N)}_{X_2|Y_1X_1UV}|$,
  \begin{itemize}
  	\item $\{u_2^j\}_{j\in \mathcal{H}^{(N)}_{X_2|Y_1X_1UV}}=w_r$, 
  	\item $\{u_2^j\}_{j\in \mathcal{H}^{(N)}_{X_2|X_1UV}\setminus \mathcal{H}^{(N)}_{X_2|Y_1X_1UV}}$ are filled with random symbols uniformly distributed over $\mathcal{X}_2$,
  	\item $\{u_2^j\}_{j\in (\mathcal{H}^{(N)}_{X_2|X_1UV})^C}$ are assigned by random mappings $\lambda_j(u_2^{1:j-1})$ that generate an output $x\in\mathcal{X}_2$ according to conditional probability $P_{U_2^{j}|X_1^{1:N}U^{1:N}V^{1:N}U_2^{1:j-1}}$,
  	\item Compute $x_2^{1:N}=u_2^{1:N}\mathbf{G}_N$. 
  \end{itemize}

  An intuitive explanation for why random symbols in $\mathcal{H}^{(N)}_{X_2|Y_1X_1UV}$ can be reused but not those in $\mathcal{H}^{(N)}_{X_2|X_1UV}\setminus \mathcal{H}^{(N)}_{X_2|Y_1X_1UV}$ is that $\{u_2^j\}_{j\in \mathcal{H}^{(N)}_{X_2|Y_1X_1UV}}$ are very unreliable for receiver 1, thus reusing them does not harm security. We will show in the next section that with such a channel prefixing approach, our proposed scheme can achieve strong secrecy. 
  
  \subsection{Decoding}
  \subsubsection{Common Message Decoding}
  
  We first consider receiver 1, who decodes from block 1 to block $m$. Although we have considered different cases in Section \ref{S:CME}, the decoding procedure can be summarized in a unified form as follows:
  
    (I) In the 1st block, receiver 1 first decodes $\{u_1^{j}\}_{j\in \mathcal{I}_{1c}^{(0)}\cup\mathcal{I}_{1c}^{(1a)}}$ with a SCD and obtains an estimate of $\bar{u}_1^{1:N}$. Then it decodes $\{u^{'j}\}_{j\in \mathcal{I}_{2c}^{(0)}\cup\mathcal{I}_{2c}^{(1a)}}$ with a SCD, in which $\bar{u}_1^{1:N}$ is treated as side information, and obtains an estimate of $\bar{u}^{'1:N}$.

  	(II) In the $i$th $(1<i<m)$ block, $\{\bar{u}_1^{j}\}_{j\in \mathcal{I}_{1c}^{(2a)}}$ and $\{\bar{u}^{'j}\}_{j\in \mathcal{I}_{2c}^{(2a)}}$ are deduced from $\bar{u}_1^{1:N}$ and $\bar{u}^{'1:N}$ in block $i-1$ according to different cases (see Fig. \ref{fig:chaining-c1} and Fig. \ref{fig:chaining-c3}), and the rest are decoded in the same way as in (I).
  	
  	(III) In the $m$th block, $\{\bar{u}_1^{j}\}_{j\in \mathcal{I}_{1c}^{(1a)}}$ and $\{\bar{u}_1^{j}\}_{j\in \mathcal{I}_{2c}^{(1a)}}$ are decoded as frozen symbols, and the rest are decoded in the same way as in (II). 
    
  Receiver 2 decodes the common messages similarly, except that it decodes from block $m$ to block 1.
  
  \subsubsection{Private and Confidential Messages Decoding}
    
  Receiver 2 decodes the private and confidential messages from block $m$ to block 1 as follows.
  
  (I) In the $m$th block, receiver 1 first decodes $\{v^{'j}\}_{j\in \mathcal{I}_{2p}\cup\mathcal{I}_{2s}}$ with $\hat{u}_1^{1:N}$ and $\hat{u}^{'1:N}$ in the same block being treated as side information, where $\hat{u}_1^{1:N}$ and $\hat{u}^{'1:N}$ are its decoding result of common messages, and obtains an estimate of $\hat{v}^{'1:N}$.
  
  (II) In the $i$th $(1\leq i<m)$ block, $\{v^{'j}\}_{i\in \mathcal{R}_2}$ are deduced from $\{\hat{v}^{'j}\}_{i\in \mathcal{I}_{2s}^{(2)}}$ in block $i+1$, and the rest are decoded in the same way as in (I).

  \section{Performance Analysis}
  \label{S:Perf}
 
  \subsection{Error Performance}
  Let $\tilde{U}^{1:N}$, $\tilde{V}^{1:N}$, $\tilde{X}_1^{1:N}$, $\tilde{X}_2^{1:N}$, $\tilde{Y}_1^{1:N}$ and $\tilde{Y}_2^{1:N}$ be the vectors generated by our encoding scheme. The following lemma shows that the joint distribution induced by our encoding scheme is asymptotically indistinguishable from the target joint distribution (the one that our scheme is designed for).
  \newtheorem{lemma}{Lemma}
  \begin{lemma}
  	\label{lemma:0}
  	\begin{align}
  	\parallel &P_{U^{1:N}V^{1:N}X_1^{1:N}X_2^{1:N}Y_1^{1:N}Y_2^{1:N}}-P_{\tilde{U}^{1:N}\tilde{V}^{1:N}\tilde{X}_1^{1:N}\tilde{X}_2^{1:N}\tilde{Y}_1^{1:N}\tilde{Y}_2^{1:N}}\parallel \nonumber\\
  	&~~~~ \leq O(\sqrt{N}2^{-N^\beta/2}).\label{TVD}
  	\end{align}
  \end{lemma}
  \begin{proof}
  	See Appendix A.
  	\end{proof}

  Denote receiver 1's error probability when decoding message $M_1^{(1)}$ in block $i$ by $P_{e1,i}^{(1)}$, and that when decoding $(M_1^{(2)},M_2^{(c)})$ by $P_{e1,i}^{(2)}$. For $i=[2,m]$, define the following error events
  \begin{align*}
  \mathcal{E}_{X_1Y_1,i}&\triangleq \{(X_1^{1:N}Y_1^{1:N})\neq (\tilde{X}_1^{1:N}\tilde{Y}_1^{1:N})_i\}, \\
  \mathcal{E}_{X_1,i-1}^{ch}&\triangleq \{(\bar{U}_1^{chaining})_{i-1}\neq (\tilde{U}_1^{chaining})_{i-1}\}, \\
  \mathcal{E}_{U,i-1}^{ch}&\triangleq \{(\bar{U}^{'chaining})_{i-1}\neq (\tilde{U}^{'chaining})_{i-1}\}, \\
  \mathcal{E}_{X_1,i}&\triangleq \{(\bar{X}_1^{1:N})_{i}\neq (\tilde{X}_1^{1:N})_{i}\}, \\
  \mathcal{E}_i&\triangleq \mathcal{E}_{X_1Y_1,i}\cup \mathcal{E}_{X_1,i-1}^{ch}\cup \mathcal{E}_{U,i-1}^{ch},\\
  \mathcal{E}'_{i}&\triangleq \mathcal{E}_{X_1Y_1,i}\cup \mathcal{E}_{X_1,i-1}^{ch}\cup \mathcal{E}_{U,i-1}^{ch}\cup \mathcal{E}_{X_1,i},
  \end{align*}
  where $(\cdot)_i$ denotes vectors in block $i$, $\bar{U}$ denotes the decoding result of $U$, and "chaining" in the superscript stands for the elements used for chaining.
  Using optimal coupling \cite[Lemma 3.6]{aldous1983random} we have
  \begin{equation*}
  P[\mathcal{E}_{X_1Y_1,i}]=\parallel P_{X_1^{1:N}Y_1^{1:N}}-P_{\tilde{X}_1^{1:N}\tilde{Y}_1^{1:N}}\parallel.
  \end{equation*}
  Then we have
  \begin{align}
 P_{e1,i}^{(1)} &\leq P[(\bar{X}_1^{1:N})_i\neq (\tilde{X}_1^{1:N})_i] \nonumber\\
  &=P[(\bar{X}_1^{1:N})_i\neq (\tilde{X}_1^{1:N})_i|\mathcal{E}_i]P[\mathcal{E}_i]\nonumber\\
  &~~~~+P[(\bar{X}_1^{1:N})_i\neq (\tilde{X}_1^{1:N})_i|\mathcal{E}_i^C]P[\mathcal{E}_i^C] \nonumber\\
  &\leq P[\mathcal{E}_i]+P[(\bar{X}_1^{1:N})_i\neq (\tilde{X}_1^{1:N})_i|\mathcal{E}_i^C] \nonumber\\
  &\leq P(\mathcal{E}_{X_1Y_1,i})+ P(\mathcal{E}_{X_1,i-1}^{ch})+ P(\mathcal{E}_{U,i-1}^{ch})\nonumber\\
  &~~~~+P[(\bar{X}_1^{1:N})_i\neq (\tilde{X}_1^{1:N})_i|\mathcal{E}_i^C] \nonumber\\
  &\leq \sqrt{2\log 2}\sqrt{N\delta_N}(2+2\sqrt{3})+N\delta_N\nonumber\\
  &~~~~+P[(\bar{X}_1^{1:N})_{i-1}\neq (\tilde{X}_1^{1:N})_{i-1}]\nonumber\\
  &~~~~~~~~ +P[(\bar{U}^{1:N})_{i-1}\neq (\tilde{U}^{1:N})_{i-1}], \label{R1C-1}
  \end{align}
  where (\ref{R1C-1}) holds from (\ref{VD-4}) and the error probability of source polar coding \cite{arikan2010source}.
  
  Similarly we have
  \begin{align}
  P_{e1,i}^{(2)}&\leq P[(\bar{U}^{1:N})_i\neq (\tilde{U}^{1:N})_i] \nonumber\\
  &\leq P(\mathcal{E}_{X_1Y_1,i})+ P(\mathcal{E}_{X_1,i-1}^{ch})+ P(\mathcal{E}_{U,i-1}^{ch})+ P(\mathcal{E}_{X_1,i})\nonumber\\
  &~~~~+P[(\bar{U}^{1:N})_i\neq (\tilde{U}^{1:N})_i|\mathcal{E}_i^{'C}] \nonumber\\
  &\leq \delta_N^{c}+N\delta_N+P[(\bar{X}_1^{1:N})_{i-1}\neq (\tilde{X}_1^{1:N})_{i-1}]\nonumber\\
  &~~~~ +P[(\bar{U}^{1:N})_{i-1}\neq (\tilde{U}^{1:N})_{i-1}]+P[(\bar{X}_1^{1:N})_{i}\neq (\tilde{X}_1^{1:N})_{i}], \label{R1C-2}
  \end{align}
  where $\delta_N^{c}\triangleq \sqrt{2\log 2}\sqrt{N\delta_N}(2+2\sqrt{3})$. From (\ref{R1C-1}) and (\ref{R1C-2}) we have
  \begin{align*}
  &~~~~P[(\bar{X}_1^{1:N})_i\neq (\tilde{X}_1^{1:N})_i]+P[(\bar{U}^{1:N})_i\neq (\tilde{U}^{1:N})_i]\\
  &\leq 3\Big{(}\delta_N^{c}+N\delta_N+P[(\bar{X}_1^{1:N})_{i-1}\neq (\tilde{X}_1^{1:N})_{i-1}]\nonumber\\
  &~~~~~~~~+P[(\bar{U}^{1:N})_{i-1}\neq (\tilde{U}^{1:N})_{i-1}]\Big{)},
  \end{align*}
  By induction we have
  \begin{align}
  &~~~~P[(\bar{X}_1^{1:N})_i\neq (\tilde{X}_1^{1:N})_i]+P[(\bar{U}^{1:N})_i\neq (\tilde{U}^{1:N})_i]\nonumber\\
  &\leq \sum_{k=1}^{i-1}3^k(\delta_N^{c}+N\delta_N)+3^{i-1}P[(\bar{X}_1^{1:N})_1\neq (\tilde{X}_1^{1:N})_1]\nonumber\\
  &~~~~~~~~+3^{i-1}P[(\bar{U}^{1:N})_1\neq (\tilde{U}^{1:N})_1].\label{R1C-3}
  \end{align}
  From the above analysis and the assumption that receivers have perfect knowledge of frozen symbols we have
  \begin{align*}
  P[(\bar{X}_1^{1:N})_1\neq (\tilde{X}_1^{1:N})_1]+P[(\bar{U}^{1:N})_1\neq (\tilde{U}^{1:N})_1]\leq 3(\delta_N^{c}+N\delta_N).
  \end{align*}
  Thus, the overall error probability of receiver 1 in a frame can be upper bounded by
  \begin{align}
  P_{e1}&\leq \sum_{k=1}^{m}\big{(}P^{(1)}_{e1,k}+P^{(2)}_{e1,k}\big{)}\nonumber\\
  &\leq \sum_{i'=1}^{m}\sum_{k=1}^{i'} 3^k(\delta_N^{c}+N\delta_N) \nonumber\\  
  &=O(3^mN2^{-N^\beta}).
  \end{align}
  
  For receiver 2, the error probability of decoding common messages in a frame can be similarly upper bounded by
  \begin{equation}
  P_{e2}^{(c)}\leq O(3^mN2^{-N^\beta}).
  \end{equation}
  To estimate receiver 2's error probability in decoding its private and confidential messages block $i$, $P_{e2,i}^{(p,s)}$, we define the following error events:
  \begin{align*}
  \mathcal{E}_{UVX_1Y_2,i}&\triangleq \{(U^{1:N}X_1^{1:N}V^{1:N}Y_2^{1:N})\neq (\tilde{U}^{1:N}\tilde{X}_1^{1:N}\tilde{V}^{1:N}\tilde{Y}_2^{1:N})_i\},\\
  \mathcal{E}_{U,i}&\triangleq \{(\bar{U}^{1:N})_i\neq (\tilde{U}^{1:N})_i\},\\
  \mathcal{E}_{X_1,i}&\triangleq \{(\bar{X}_1^{1:N})_i\neq (\tilde{X}_1^{1:N})_i\},\\
  \mathcal{E}_{V,i+1}&\triangleq \{(\bar{V}^{1:N})_{i+1}\neq (\tilde{V}^{1:N})_{i+1}\},\\
  \mathcal{E}_i&\triangleq \mathcal{E}_{UVX_1Y_2,i}\cup \mathcal{E}_{U,i}\cup \mathcal{E}_{X_1,i}\cup \mathcal{E}_{V,i+1}.
  \end{align*}
  Using optimal coupling \cite[Lemma 3.6]{aldous1983random} we have
  \begin{equation*}
  P[\mathcal{E}_{UVX_1Y_2,i}]=\parallel P_{U^{1:N}V^{1:N}X_1^{1:N}Y_1^{1:N}}-P_{\tilde{U}^{1:N}\tilde{V}^{1:N}\tilde{X}_1^{1:N}\tilde{Y}_1^{1:N}}\parallel.
  \end{equation*}
  
 Similar to the analysis for common message decoding, $P_{e2,i}^{(p,s)}$ can be upper bounded by
  \begin{align*}
  P_{e2,i}^{(p,s)}&\leq P[\mathcal{E}_{V,i}]\\
  &\leq \delta_N^{c}+N\delta_N+P[\mathcal{E}_{U,i}]+P[ \mathcal{E}_{X_1,i}]+P[\mathcal{E}_{V,i+1}]\\
  &\leq \sum_{k=i}^{m} (3^{m-k}+1)(\delta_N^{c}+N\delta_N)+P[\mathcal{E}_{V,i+1}],
  \end{align*}
  where $\delta_N^{c}\triangleq \sqrt{2\log 2}\sqrt{N\delta_N}(2+2\sqrt{3})$. By induction and (\ref{R1C-3}) we have
  \begin{align}
  P_{e2,i}^{(p,s)}  &\leq \sum_{i'=i}^{m}\sum_{k=i'}^{m} (3^{m-k}+1)(\delta_N^{c}+N\delta_N) ).
  \end{align}
  Then
  \begin{align}
  P_{e2}^{(p,s)}\leq \sum_{k=1}^{m}P_{e2,k}^{(p,s)}=O(3^mN2^{-N^\beta}).
  \end{align}
  
  \subsection{Secrecy}
  \label{S:Secrecy}
  
  We first introduce some notations used in this subsection. In the $i$th $(1\leq i \leq m)$ block, the outputs of Enc 1, 2a and 2b (see Fig. \ref{fig:encoding}) are denoted by $\mathbf{X}_{1,i}$, $\mathbf{U}_i$ and $\mathbf{V}_i$, respectively. Transmitter 2's confidential message at $\mathcal{I}_{2s}^{(1)}$ is denoted by $M_i$, and private message at $\mathcal{I}_{2s}^{(2)}$ (which is used for chaining) by $E_i$. Receiver 1's channel output is denoted by $\mathbf{Y}_{1,i}$. The additionally transmitted almost deterministic symbols in $U_1^{1:N}$ and $U^{'1:N}$ are denoted by $D_{1c,i}$ and $D_{2c,i}$, respectively. The reused frozen symbols in $U_1^{1:N}$, $U^{'1:N}$ and $V^{'1:N}$ are denoted by $F_{1c}$, $F_{2c}$ and $F_{2p}$, respectively. The non-reused frozen symbols (see Fig. \ref{fig:chaining-c1} and \ref{fig:chaining-c3}) in $U_1^{1:N}$ in the 1st and $m$th blocks are denoted by $F_{11}$ and $F_{1m}$ respectively, and those in $U^{'1:N}$ by $F_{21}$ and $F_{2m}$ respectively. The reused randomness at $\mathcal{H}^{(N)}_{X_2|Y_1X_1UV}$ in the channel prefixing scheme is denoted by $W$. For brevity, denote $F\triangleq \{F_{1c},F_{2c},F_{11},F_{1m},F_{21},F_{2m},F_{2p}\}$, $D_i\triangleq \{D_{1c,i},D_{2c,i}\}$, $M^{i:m}\triangleq \{M_i,...,M_m\}$, etc.
  
  \begin{lemma}
  	\label{lemma:}
  	For any $i\in [m]$, we have
  	\begin{equation}
  	I(M_i,E_i;\mathbf{Y}_{1,i},D_i,F)\leq O(N^3 2^{-N^{\beta}}).
  	\end{equation}
  \end{lemma}
  \begin{proof}
  	See Appendix B.
  	
  \end{proof}
  
  \begin{lemma}
  	\label{lemma:2}
  	For any $i\in [1,m-1]$,
  	\begin{align}
  	&I(W;\mathbf{Y}_1^{i:m},D^{i:m},F|M^{i:m},E_i)\nonumber\\
  	&-I(E_{i+1},W;\mathbf{Y}_1^{i+1:m},D^{i+1:m},F|M^{i+1:m})\leq O(N^3 2^{-N^{\beta}}).
  	\end{align}
  \end{lemma}
  \begin{proof}
  	See Appendix C.
  \end{proof}

  \begin{lemma}
  	\label{lemma:3}
  	For any $i\in [1,m-1]$, let
  	\begin{equation}
  	L_i=I(M^{i:m},E_i,W;\mathbf{Y}_1^{i:m},D^{i:m},F).
  	\end{equation}
  	Then we have
  	\begin{equation}
  	L_i-L_{i+1}\leq O(N^3 2^{-N^{\beta}}).
  	\end{equation}
  \end{lemma}
  \begin{proof}
  	See Appendix D.  	
  \end{proof}

  Suppose receiver 1 has perfect knowledge of the frozen symbols in each block. Then the information leakage is
  \begin{align}
  L(N)&=I(M^{1:m};\mathbf{Y}_1^{1:m},D^{1:m},F)\nonumber\\
  &\leq I(M^{1:m},\textsc{E}_m,W;\mathbf{Y}_1^{1:m},D^{1:m},F).\nonumber
  \end{align}
  From the proof of Lemma \ref{lemma:3} and the fact that receiver 1 has no knowledge about the secret seed we have $L_m\leq O(N^3 2^{-N^{\beta}})$.
  Thus, by induction hypothesis we have
  \begin{align}
  L(N)&\leq \sum_{i=1}^{m-1}\big{(} L_i-L_{i+1} \big{)}+L_m \leq O(mN^3 2^{-N^{\beta}}).
  \end{align}

  \subsection{Achievable Rate Region}
  \label{S:ARR}
  
  \subsubsection{Randomness Rate}
  
  Since $w_r$ is reused in a frame, the generating rate of randomness required by our channel prefixing scheme is
  \begin{equation}
  R_r=\frac{1}{mN}\big{(} |\mathcal{H}^{(N)}_{X_2|Y_1X_1UV}|+m|\mathcal{H}^{(N)}_{X_2|X_1UV}\setminus \mathcal{H}^{(N)}_{X_2|Y_1X_1UV}|\big{)}.
  \end{equation}
  From \cite[Lemma 1]{chou2015keygen} we have
  \begin{equation}
  \lim\limits_{N\rightarrow\infty}\frac{1}{N}|\mathcal{H}^{(N)}_{X_2|Y_1X_1UV})^C\setminus \mathcal{L}^{(N)}_{X_2|Y_1X_1UV}|=0.
  \end{equation}
  Then we have
  \begin{align}
  \lim\limits_{N\rightarrow\infty, m\rightarrow\infty}R_r&=\lim\limits_{N\rightarrow\infty}\frac{1}{N}|\mathcal{H}^{(N)}_{X_2|X_1UV}\cap (\mathcal{H}^{(N)}_{X_2|Y_1X_1UV})^C|\nonumber\\
  &=\lim\limits_{N\rightarrow\infty}\frac{1}{N}|\mathcal{H}^{(N)}_{X_2|X_1UV}\cap \mathcal{L}^{(N)}_{X_2|Y_1X_1UV}|\nonumber\\
  &=I(X_2;Y_1|U,V,X_1).\label{AsympRr}
  \end{align}
  As has been noted in \cite{Watanabe2014cognitive}, the difference between transmitter 2's private message and the randomness required by the encoder is just whether it carries information. We can see that (\ref{AsympRr}) is the minimum generating rate of randomness required. Thus, (\ref{ineq-6}) is achievable with our proposed scheme.
  
  \subsubsection{Private and Confidential Message Rates}
  
  From Fig. \ref{fig:chaining-P} we can see that the private and confidential message rates in our proposed scheme are
  \begin{align}
  R_2^{(p)}=\frac{1}{N}\big{(}|\mathcal{I}_{2p}|+|\mathcal{R}_2|\big{)},~~~
  R_2^{(s)}=\frac{1}{N}|\mathcal{I}_{2s}^{(1)}|,
  \end{align}
  respectively. By a similar analysis to the general wiretap polar code \cite{gulcu2017wiretap}, we have
  \begin{align}
  \lim\limits_{N\rightarrow \infty} R_2^{(p)}&=I(V;Y_1|U,X_1),\label{Rp-1}\\
  \lim\limits_{N\rightarrow \infty} R_2^{(s)}&=I(V;Y_2|U,X_1)-I(V;Y_1|U,X_1). \label{Rs-1}
  \end{align}
  Thus, (\ref{ineq-4}) can be achieved. 
  
  In the private and confidential message encoding procedure introduced in Section \ref{S:PCME}, positions in $V^{'1:N}$ allocated for private messages can also be assigned with randomness. And note that we can allocate more randomness in the encoder without sacrificing reliability or secrecy (e.g., we can replace some of the confidential message symbols with random symbols), but we can never allocate less of them.  Thus, from (\ref{AsympRr}) and (\ref{Rp-1}) we can see that (\ref{ineq-5}) can be achieved.

  \subsubsection{Common Message Rate}
  
  In Case 1, the common message rates of two transmitters in our proposed scheme are
  \begin{equation}
  \label{CMR}
  \begin{aligned}
  R_1^{(1)}&=\frac{m|\mathcal{I}_{1c}^{(0)}|+(m-1)|\mathcal{I}_{1c}^{(1a)}|}{mN}=\frac{|\mathcal{I}_{1c}^{(1)}|}{N}-\frac{|\mathcal{I}_{1c}^{(1a)}|}{mN},\\
  R^{(c)}&=\frac{m|\mathcal{I}_{2c}^{(0)}|+(m-1)|\mathcal{I}_{2c}^{(1a)}|}{mN}=\frac{|\mathcal{I}_{2c}^{(1)}|}{N}-\frac{|\mathcal{I}_{2c}^{(1a)}|}{mN},
  \end{aligned}
  \end{equation}
  respectively. From (\ref{CommonRate}) we have
  \begin{equation}
  \lim\limits_{N\rightarrow \infty,m\rightarrow \infty} R_1^{(1)}=I(X_1;Y_1),~~
  \lim\limits_{N\rightarrow \infty,m\rightarrow \infty} R^{(c)}=I(U;Y_1|X_1).  \label{C1-Rc}
  \end{equation}
  Since $\min\{I(U,X_1;Y_1),I(U,X_1;Y_2)\}=I(U,X_1;Y_1)$ in this case, if we allocate all of $R^{(c)}$ to transmitter 1's message, then (\ref{ineq-1}) is achieved. No matter how we allocate two transmitters' common messages, the sum rate of all messages always achieves (\ref{ineq-3}). 
  
  To prove the achievability of (\ref{ineq-2}) in Case 1 requires some change in the coding scheme. If we wish to maximize the sum rate of private and confidential messages, transmitter 2 will not help transmit $M_1$ at all. Therefore, whether receiver 1 can decode $U$ does not matter. Then transmitter 2 can use all of $\mathcal{I}_{2c}^{(2)}$ to transmit its own message at any rate below $I(U;Y_2|X_1)$ (now this message becomes private message). Then from (\ref{Rp-1}) and (\ref{Rs-1}) we can see that (\ref{ineq-2}) is achieved. 
  
  Another thing worth noting is that, in the above case, in order for both receivers to decode transmitter 1's message, $R_1\leq \min\{I(X_1;Y_1),I(X_1;Y_2)\}$ must hold. Then the sum rate of all messages satisfies
  \begin{align*}
  R_1+R_{2p}+R_{2s}\leq \min\{I(X_1;Y_1),I(X_1;Y_2)\}+I(U,V;Y_2|X_1),
  \end{align*}
  which may seem to violate (\ref{ineq-3}). However, since $U$ in fact carries private message in this case, it is equivalent to remove auxiliary random variable $U$ and simply design a code on $V$. We can also see this problem from the mutual information aspect. Due to the Markov chains of (\ref{constraint-1}) and (\ref{constraint-2}), we have
  \begin{align*}
  I(U,V;Y_2|X_1)&=I(V;Y_2|X_1)+I(U;Y_2|V,X_1)\\
  &=I(V;Y_2|X_1).
  \end{align*}
  With auxiliary random variable $U$ being removed, we can readily see that (\ref{ineq-3}) still holds.
  
  In Case 3-1, $R_1^{(1)}$ and $R^{(c)}$ are the same as in Case 1, thus (\ref{ineq-1}) and (\ref{ineq-3}) are achievable. As we have explained in Section \ref{S:Case34}, $\{u^{'j}\}_{j\in \mathcal{I}_{2c}^{(1b)}}$ must be assigned to transmitter 1's common message. Thus, in this case,
  \begin{equation}
  \label{Rc-2}
  \begin{aligned}
  R_2^{(c)}&\leq I(U;Y_1|X_1)-\big{(} I(U;Y_1|X_1)-I(U;Y_2|X_1) \big{)}\\
  &=I(U;Y_2|X_1).
  \end{aligned}
  \end{equation}
  Then from (\ref{Rp-1}), (\ref{Rs-1}) and (\ref{Rc-2}) we can see that (\ref{ineq-2}) is achieved ($R_{2p}$ in Theorem \ref{Capacity-2} is the sum of $R_2^{(c)}$ and $R_2^{(p)}$ here).
  
  Since Case 2 (resp. 4) is similar to Case 1 (resp. 3), and Case 3-2 is similar to Case 3-1, we can now conclude that our proposed scheme achieves the whole region in Theorem \ref{Capacity-2} under the strong secrecy criterion with randomness constraint.

 \section{Discussion}
 \label{S:Conclusion}
 
  Although our proposed polar coding scheme is designed under secrecy constraints, it can be readily modified for the case without secrecy and achieve the capacity region of the CIC given by \cite[Theorem 4]{liang2009cognitive}, since the capacity region is just a special case of the capacity-equivocation region when secrecy constraints are removed.
  
  We note some relations between our work and \cite{chou2016broadcast}, which considers polar coding for the broadcast channel with confidential messages. From Theorem \ref{Capacity-2} and \cite[Theorem 1]{chou2016broadcast} we can see that the rate region in \cite[Theorem 1]{chou2016broadcast} is a special case of that in Theorem \ref{Capacity-2}  when transmitter 1 is removed. Also, as shown in \cite{liang2009cognitive}, the region defined in Theorem \ref{Capacity-1} reduces to the capacity region of the MAC with degraded message sets if we set $Y_1=Y_2$. Thus, our proposed scheme can be seen as a general solution for the aforementioned multi-user polar coding problems.
 
  %A practical issue of our proposed scheme (as well as other polar coding schemes for general multi-user channels) might be the chaining method we used, which may cause large decoding latency. The chaining method is a common solution when the channels observed by different receivers do not exhibit a degradation relation. However, in many practical situations, such as the Gaussian CICC, the degradation relation often holds. Thus, low-latency single-block schemes are possible in practice\footnote{Under the weak secrecy criterion, secrecy capacity achieving single-block polar coding schemes can be readily designed if the degradation relation holds. However, if we wish to achieve strong secrecy within a single transmission block, extra requirements (such as the legitimate users share some randomness secretly) will be needed.}. In the Gaussian CICC, to guarantee the encoder outputs are also Gaussian distributed, one may apply polar lattice coding in our proposed scheme \cite{yan2014polarlattice,liu2015lattice}.
 
 \section*{APPENDIX A\\Proof of Lemma 1}
 Similar to the proof of \cite[Lemma 5]{chou2016broadcast}, we have
 \begin{align*}
 \mathbb{D}(P_{X_1^{1:N}}||P_{\tilde{X}_1^{1:N}})&\leq N\delta_N,\\
 \mathbb{D}(P_{U^{1:N}X_1^{1:N}}||P_{\tilde{U}^{1:N}\tilde{X}_1^{1:N}})&\leq 2N\delta_N,\\
 \mathbb{D}(P_{U^{1:N}V^{1:N}X_1^{1:N}}||P_{\tilde{U}^{1:N}\tilde{V}^{1:N}\tilde{X}_1^{1:N}})&\leq 3N\delta_N,\\
 \mathbb{D}(P_{X_2^{1:N}V^{1:N}}||P_{\tilde{X}_2^{1:N}\tilde{V}^{1:N}})&\leq 4N\delta_N.
 \end{align*}
 Then we have
 \begin{align}
 &~~\parallel P_{U^{1:N}V^{1:N}X_1^{1:N}X_2^{1:N}Y_1^{1:N}Y_2^{1:N}}-P_{\tilde{U}^{1:N}\tilde{V}^{1:N}\tilde{X}_1^{1:N}\tilde{X}_2^{1:N}\tilde{Y}_1^{1:N}\tilde{Y}_2^{1:N}}\parallel \nonumber\\
 &= \parallel P_{X_2^{1:N}|V^{1:N}}P_{U^{1:N}V^{1:N}X_1^{1:N}}-P_{\tilde{X}_2^{1:N}|\tilde{V}^{1:N}}P_{\tilde{U}^{1:N}\tilde{V}^{1:N}\tilde{X}_1^{1:N}}\parallel \label{VD-1}\\
 &\leq \parallel P_{X_2^{1:N}|V^{1:N}}P_{U^{1:N}V^{1:N}X_1^{1:N}}-P_{\tilde{X}_2^{1:N}|\tilde{V}^{1:N}}P_{U^{1:N}V^{1:N}X_1^{1:N}}\parallel \nonumber\\
 &~~~~+ \parallel P_{\tilde{X}_2^{1:N}|\tilde{V}^{1:N}}P_{U^{1:N}V^{1:N}X_1^{1:N}}-P_{\tilde{X}_2^{1:N}|\tilde{V}^{1:N}}P_{\tilde{U}^{1:N}\tilde{V}^{1:N}\tilde{X}_1^{1:N}}\parallel \label{VD-2}\\
 &= \parallel P_{X_2^{1:N}|V^{1:N}}P_{V^{1:N}}-P_{\tilde{X}_2^{1:N}|\tilde{V}^{1:N}}P_{V^{1:N}}\parallel \nonumber\\
 &~~~~+ \parallel P_{U^{1:N}V^{1:N}X_1^{1:N}}-P_{\tilde{U}^{1:N}\tilde{V}^{1:N}\tilde{X}_1^{1:N}}\parallel \label{VD-2-1}\\
 &\leq \parallel P_{X_2^{1:N}V^{1:N}}-P_{\tilde{X}_2^{1:N}\tilde{V}^{1:N}}\parallel \nonumber\\
 &~~+\parallel P_{\tilde{X}_2^{1:N}\tilde{V}^{1:N}}-P_{\tilde{X}_2^{1:N}|\tilde{V}^{1:N}}P_{V^{1:N}}\parallel \nonumber\\
 &~~~~+ \parallel P_{U^{1:N}V^{1:N}X_1^{1:N}}-P_{\tilde{U}^{1:N}\tilde{V}^{1:N}\tilde{X}_1^{1:N}}\parallel \label{VD-3}\\
 &\leq \parallel P_{X_2^{1:N}V^{1:N}}-P_{\tilde{X}_2^{1:N}\tilde{V}^{1:N}}\parallel +\parallel P_{V^{1:N}}-P_{\tilde{V}^{1:N}}\parallel \nonumber\\
 &~~~~+ \parallel P_{U^{1:N}V^{1:N}X_1^{1:N}}-P_{\tilde{U}^{1:N}\tilde{V}^{1:N}\tilde{X}_1^{1:N}}\parallel \label{VD-3-1}\\
 &\leq \sqrt{2\log 2}\sqrt{N\delta_N}(2+2\sqrt{3}) \label{VD-4}\\
 &= O(\sqrt{N}2^{-N^\beta/2}), \nonumber
 \end{align}
 where (\ref{VD-1}), (\ref{VD-2-1}) and (\ref{VD-3-1}) hold by \cite[Lemma 17]{cuff2009communication}, and (\ref{VD-2}) and (\ref{VD-3}) hold by the triangle inequality.

 \section*{APPENDIX B\\Proof of Lemma 2}
 Let $t=|\mathcal{I}_{2s}|+|\mathcal{I}_{2p}|$ and $w=|\mathcal{F}_2|$. Denote $\{a_1,a_2,...,a_t\}=\mathcal{I}_{2s}$ with $a_1<...<a_t$, $\{b_1,b_2,...,b_w\}=\mathcal{F}_2$ with $b_1<...<b_w$, and $\{c_1,c_2,...,c_{t+w}\}=\{a_1,...,a_t,b_1,...,b_w\}$ with $c_1<...<c_{t+w}$. Let $F_c$ be short for $\{F_{1c},F_{2c},F_{11},F_{1m},F_{21},F_{2m}\}$. Then we have
 \begin{align}
 &~~~~I(M_i,E_i;\mathbf{Y}_{1,i},D_i,F)\nonumber\\
 &= H_{q_V}(M_i,E_i)-H_{q_V}(M_i,E_i|\mathbf{Y}_{1,i},D_i,F)\nonumber\\
 &=H_{q_V}(M_i,E_i)-H_{q_V}(M_i,E_i,F_{2p}|\mathbf{Y}_{1,i},D_i,F_c)\nonumber\\
 &~~~~+H_{q_V}(F_{2p}|\mathbf{Y}_{1,i},D_i,F_c)\nonumber\\
 &\leq t+w-H_{q_V}(M_i,E_i,F_{2p}|\mathbf{Y}_{1,i},\mathbf{X}_{1,i},\mathbf{U}_i)\label{lemma1-1}\\
 &= t+w-\sum_{j=1}^{t+w}H_{q_V}(\tilde{V}^{'c_j}|\tilde{Y}_{1}^{1:N},\tilde{X}_{1}^{1:N},\tilde{U}^{1:N},\tilde{V}^{'\{c_1,...,c_{j-1}\}})\nonumber\\
 &\leq \sum_{j=1}^{t+w}\big{(} 1-H_{q_V}(\tilde{V}^{'c_j}|\tilde{Y}_{1}^{1:N},\tilde{X}_{1}^{1:N},\tilde{U}^{1:N},\tilde{V}^{'1:c_{j-1}}) \big{)},\label{lemma1-2}
 \end{align}
 where (\ref{lemma1-1}) holds because $$H_{q_V}(M_i,E_i)\leq t,~~H_{q_V}(F_{2p}|\mathbf{Y}_{1,i},D_i,F_c)\leq w,$$ and 
 \begin{equation}
 H_{q_V}(M_i,E_i,F_{2p}|\mathbf{Y}_{1,i},D_i,F_c)\geq H_{q_V}(M_i,E_i,F_{2p}|\mathbf{Y}_{1,i},\mathbf{X}_{1,i},\mathbf{U}_i), \label{APA-1} \end{equation}
 which is shown in more details as follows. For $i=1$, 
 \begin{align*}
 &~~~~H_{q_V}(M_i,E_i,F_{2p}|\mathbf{Y}_{1,i},D_i,F_c)\\
 &=H_{q_V}(M_i,E_i,F_{2p}|\mathbf{Y}_{1,i},D_i,F_{1c},F_{2c},F_{11},F_{21})
 \end{align*}
 because $(F_{1m},F_{2m})$ is independent from the rest items in the left-hand-side of (\ref{APA-1}). Thus, (\ref{APA-1}) holds. For $i=m$ and $1<i<m$, we can similarly show that (\ref{APA-1}) always holds.
 
 Note that the entropies above are calculated under the induced distribution by the encoding scheme. Under the target distribution $P_{U^{1:N}V^{1:N}X_1^{1:N}Y_1^{1:N}}$, from (\ref{Construct}) we have
 \begin{align}
 H_{q_V}(V^{'c_j}|Y_{1}^{1:N},X_{1}^{1:N},U^{1:N},V^{'1:c_{j-1}})\geq 1-\delta_N.\label{lemma1-3}
 \end{align}
 From \cite[Theorem 17.3.3]{cover2012informtaion} we have
 \begin{align}
 &~~~~|H_{q_V}(\tilde{V}^{'c_j}|\tilde{Y}_{1}^{1:N},\tilde{X}_{1}^{1:N},\tilde{U}^{1:N},\tilde{V}^{'1:c_{j-1}})\nonumber\\
 &~~~~-H_{q_V}(V^{'c_j}|Y_{1}^{1:N},X_{1}^{1:N},U^{1:N},V^{'1:c_{j-1}})|\nonumber\\
 &\leq \parallel P_{U^{1:N}V^{1:N}X_1^{1:N}Y_1^{1:N}}-P_{\tilde{U}^{1:N}\tilde{V}^{1:N}\tilde{X}_1^{1:N}\tilde{Y}_1^{1:N}}\parallel \nonumber\\
 &~~~~\times \log\frac{|\mathcal{U}|^{N}|\mathcal{V}|^{N}|\mathcal{X}_1|^N|\mathcal{Y}_1|^N}{\parallel P_{U^{1:N}V^{1:N}X_1^{1:N}Y_1^{1:N}}-P_{\tilde{U}^{1:N}\tilde{V}^{1:N}\tilde{X}_1^{1:N}\tilde{Y}_1^{1:N}}\parallel}\nonumber\\
 &=O(N^2 2^{-N^{\beta}})+O(N^{\beta+1}2^{-N^{\beta}}).\label{Lemma1-2}
 \end{align}
 From (\ref{lemma1-2}) and (\ref{lemma1-3}) we have
 \begin{equation*}
 H_{q_V}(V^{'c_j}|Y_{1}^{1:N},X_{1}^{1:N},U^{1:N},V^{'1:c_{j-1}})\geq 1-O(N^2 2^{-N^{\beta}}).
 \end{equation*}
 Thus,
 \begin{equation}
 I(M_i,E_i;\mathbf{Y}_{1,i},D_i,F)\leq O(N^3 2^{-N^{\beta}}).
 \end{equation}	
 
 \section*{APPENDIX C\\Proof of Lemma 3}
 To prove Lemma \ref{lemma:2}, we first prove the following two lemmas.
 \begin{lemma}
 	\label{lemma:4}
 	For any $i\in [1,m]$,
 	\begin{align}
 	I(E_{i};W|\mathbf{Y}_1^{i:m},D^{i:m},M^{i:m})\leq O(N^3 2^{-N^{\beta}}). \nonumber
 	\end{align}
 \end{lemma}
 \begin{proof}
 	Since $E_{i}$ is independent from $(\mathbf{Y}_1^{i+1:m},D^{i+1:m},M^{i+1:m})$, we have
 	\begin{align}
 	&~~~~I(E_{i};W|\mathbf{Y}_1^{i:m},D^{i:m},M^{i:m}) \nonumber\\
 	&=I(E_{i};W|\mathbf{Y}_{1,i},D_{i},M_{i})\nonumber\\
 	&=H_{q_{X_2}}(W|\mathbf{Y}_{1,i},D_{i},M_{i})-H_{q_{X_2}}(W|\mathbf{Y}_{1,i},D_{i},M_{i},E_{i})\nonumber\\
 	&\leq H_{q_{X_2}}(W)-H_{q_{X_2}}(W|\mathbf{Y}_{1,i},\mathbf{X}_{1,i},\mathbf{U}_i,\mathbf{V}_i). \nonumber
 	\end{align}
 	Then we can prove this lemma similar to the proof of Lemma \ref{lemma:1}.
 	
\end{proof}	
 
 \begin{lemma}
 	\label{lemma:5}
 	For any $i\in [1,m-1]$,
 	\begin{align}
 	&I(\mathbf{Y}_{1,i},D_i,F,M_i,E_i;\mathbf{Y}_1^{i+1:m},D^{i+1:m},M^{i+1:m}|W)\nonumber\\
 	&~~~~ -I(E_{i+1},W;\mathbf{Y}_1^{i+1:m},D^{i+1:m},F,M^{i+1:m})\nonumber\\
 	&~~~~~~- I(\mathbf{Y}_{1,i},D_i,F,M_i,E_i;\mathbf{Y}_1^{i+1:m},D^{i+1:m},M^{i+1:m})\nonumber\\
 	&~~~~~~~~+I(W;\mathbf{Y}_1^{i+1:m},D^{i+1:m},M^{i+1:m})\nonumber\\
 	&\leq O(N^3 2^{-N^{\beta}}). \nonumber
 	\end{align}
 \end{lemma}
 \begin{proof}
 	\begin{align}
 	&~~~~I(\mathbf{Y}_{1,i},D_i,F,M_i,E_i;\mathbf{Y}_1^{i+1:m},D^{i+1:m},M^{i+1:m}|W) \nonumber\\
 	&~~~~~~-I(E_{i+1},W;\mathbf{Y}_1^{i+1:m},D^{i+1:m},F,M^{i+1:m})\nonumber\\
 	&\leq I(\mathbf{Y}_{1,i},D_i,F,M_i,E_i,E_{i+1};\mathbf{Y}_1^{i+1:m},D^{i+1:m},M^{i+1:m}|W)\nonumber\\
 	&~~~~ -I(E_{i+1},W;\mathbf{Y}_1^{i+1:m},D^{i+1:m},F,M^{i+1:m})\nonumber\\
 	&=I(\mathbf{Y}_{1,i},D_i,M_i,E_i;\mathbf{Y}_1^{i+1:m},D^{i+1:m},M^{i+1:m}|W,F,E_{i+1})\nonumber\\
 	&~~~~+I(F,E_{i+1};\mathbf{Y}_1^{i+1:m},D^{i+1:m},M^{i+1:m}|W)\nonumber\\
 	&~~~~~~-I(E_{i+1},W;\mathbf{Y}_1^{i+1:m},D^{i+1:m},F,M^{i+1:m})\nonumber\\
 	&=I(F,E_{i+1};\mathbf{Y}_1^{i+1:m},D^{i+1:m},M^{i+1:m},W)\nonumber\\
 	&~~~~-I(E_{i+1},W;\mathbf{Y}_1^{i+1:m},D^{i+1:m},F,M^{i+1:m})\label{lemma5-1}\\
 	&=I(F,E_{i+1};\mathbf{Y}_1^{i+1:m},D^{i+1:m},M^{i+1:m})\nonumber\\
 	&~~~~+I(F,E_{i+1};W|\mathbf{Y}_1^{i+1:m},D^{i+1:m},M^{i+1:m})\nonumber\\
 	&~~~~~~-I(E_{i+1},W;\mathbf{Y}_1^{i+1:m},D^{i+1:m},M^{i+1:m})\nonumber\\
 	&~~~~~~~~-I(E_{i+1},W;F|\mathbf{Y}_1^{i+1:m},D^{i+1:m},M^{i+1:m})\nonumber\\
 	&=I(F;\mathbf{Y}_1^{i+1:m},D^{i+1:m},M^{i+1:m}|E_{i+1})\nonumber\\
 	&~~~~-I(W;\mathbf{Y}_1^{i+1:m},D^{i+1:m},M^{i+1:m}|E_{i+1}) \nonumber\\
 	&~~~~~~+I(F,E_{i+1};W|\mathbf{Y}_1^{i+1:m},D^{i+1:m},M^{i+1:m})\nonumber\\
 	&~~~~~~~~-I(E_{i+1},W;F|\mathbf{Y}_1^{i+1:m},D^{i+1:m},M^{i+1:m})\nonumber\\
 	&=I(F;\mathbf{Y}_1^{i+1:m},D^{i+1:m},M^{i+1:m},E_{i+1})\nonumber\\
 	&~~~~-I(W;\mathbf{Y}_1^{i+1:m},D^{i+1:m},M^{i+1:m},E_{i+1}) \nonumber\\
 	&~~~~~~+I(F,E_{i+1};W|\mathbf{Y}_1^{i+1:m},D^{i+1:m},M^{i+1:m})\nonumber\\
 	&~~~~~~~~-I(E_{i+1},W;F|\mathbf{Y}_1^{i+1:m},D^{i+1:m},M^{i+1:m}) \label{lemma5-2}\\
 	&=I(F;\mathbf{Y}_1^{i+1:m},D^{i+1:m},M^{i+1:m})\nonumber\\
 	&~~~~+I(F;E_{i+1}|\mathbf{Y}_1^{i+1:m},D^{i+1:m},M^{i+1:m})\nonumber\\
 	&~~~~~~ -I(W;\mathbf{Y}_1^{i+1:m},D^{i+1:m},M^{i+1:m},E_{i+1})\nonumber\\
 	&~~~~~~~~+I(F,E_{i+1};W|\mathbf{Y}_1^{i+1:m},D^{i+1:m},M^{i+1:m})\nonumber\\
 	&~~~~~~~~~~-I(E_{i+1},W;F|\mathbf{Y}_1^{i+1:m},D^{i+1:m},M^{i+1:m})\nonumber\\
 	&= I(F;\mathbf{Y}_1^{i+1:m},D^{i+1:m},M^{i+1:m})\nonumber\\
 	&~~~~-I(W;\mathbf{Y}_1^{i+1:m},D^{i+1:m},M^{i+1:m},E_{i+1})\nonumber\\
 	&~~~~~~+I(E_{i+1};W|\mathbf{Y}_1^{i+1:m},D^{i+1:m},M^{i+1:m}),\nonumber
 	\end{align}
 	where (\ref{lemma5-1}) holds because block $i$ and the next $m-i$ blocks are independent conditioned on $(W,F,E_{i+1})$ and $W$ is independent from $(F,E_{i+1})$, and (\ref{lemma5-2}) holds due to the independence between $E_{i+1}$ and $(F,W)$. Since \begin{align*}
 	&~~~~I(F;\mathbf{Y}_1^{i+1:m},D^{i+1:m},M^{i+1:m})\\
 	&\leq I(\mathbf{Y}_{1,i},D_i,F,M_i,E_i;\mathbf{Y}_1^{i+1:m},D^{i+1:m},M^{i+1:m})
 	\end{align*}
 	and 
 	\begin{align*}
 	&~~~~I(W;\mathbf{Y}_1^{i+1:m},D^{i+1:m},M^{i+1:m},E_{i+1})\\
 	&\geq I(W;\mathbf{Y}_1^{i+1:m},D^{i+1:m},M^{i+1:m}),
 	\end{align*}
 	by Lemma \ref{lemma:4} we can readily prove this lemma.
 \end{proof}
 
  Now we will prove Lemma \ref{lemma:2}. Since $W$, $M_i$ and $E_i$ are independent from one another, we have
  \begin{align}
  &~~~~I(W;\mathbf{Y}_1^{i:m},D^{i:m},F|M^{i:m},E_i)\nonumber\\
  &~~~~-I(E_{i+1},W;\mathbf{Y}_1^{i+1:m},D^{i+1:m},F|M^{i+1:m})\nonumber\\
  &=I(W;\mathbf{Y}_1^{i:m},D^{i:m},F,M^{i:m},E_i)\nonumber\\
  &~~~~-I(E_{i+1},W;\mathbf{Y}_1^{i+1:m},D^{i+1:m},F,M^{i+1:m}) \nonumber\\  	
  &= I(W;\mathbf{Y}_{1,i},D_i,F,M_i,E_i)\nonumber\\
  &~~~~+I(W;\mathbf{Y}_1^{i+1:m},D^{i+1:m},M^{i+1:m}|\mathbf{Y}_{1,i},D_i,F,M_i,E_i) \nonumber\\
  &~~~~~~-I(E_{i+1},W;\mathbf{Y}_1^{i+1:m},D^{i+1:m},F,M^{i+1:m})\nonumber\\  	
  &=I(W;\mathbf{Y}_{1,i},D_i,F,M_i,E_i)+I(W;\mathbf{Y}_1^{i+1:m},D^{i+1:m},M^{i+1:m}) \nonumber\\  
  &~~~~+I(\mathbf{Y}_{1,i},D_i,F,M_i,E_i;\mathbf{Y}_1^{i+1:m},D^{i+1:m},M^{i+1:m}|W) \nonumber\\  
  &~~~~~~-I(\mathbf{Y}_{1,i},D_i,F,M_i,E_i;\mathbf{Y}_1^{i+1:m},D^{i+1:m},M^{i+1:m}) \nonumber\\
  &~~~~~~~~-I(E_{i+1},W;\mathbf{Y}_1^{i+1:m},D^{i+1:m},F,M^{i+1:m})\nonumber\\
  &\leq I(W;\mathbf{Y}_{1,i},D_i,F,M_i,E_i)+O(N^3 2^{-N^{\beta}}), \label{lemma2-1}
  \end{align}
  where (\ref{lemma2-1}) holds due to Lemma \ref{lemma:5}.
  Similarly to the proof of Lemma \ref{lemma:1}, we can show that
  \begin{align*}
  I(W;\mathbf{Y}_{1,i},D_i,F,M_i,E_i)  &\leq I(W;\mathbf{Y}_{1,i},\mathbf{X}_{1,i},\mathbf{U}_i,\mathbf{V}_i)\\
  &\leq O(N^3 2^{-N^{\beta}}).
  \end{align*}
  Then we can see that Lemma \ref{lemma:2} holds.

  \section*{APPENDIX D\\Proof of Lemma 4}
  \begin{align}
  &~~~~I(M^{i:m},E_i,W;\mathbf{Y}_1^{i:m},D^{i:m},F)\nonumber\\
  &~~~~-I(M^{i+1:m},E_{i+1},W;\mathbf{Y}_1^{i+1:m},D^{i+1:m},F)\nonumber\\
  &= I(M_i,E_i,W;\mathbf{Y}_1^{i:m},D^{i:m},F|M^{i+1:m})\nonumber\\
  &~~~~+I(M^{i+1:m};\mathbf{Y}_{1,i},D_i|\mathbf{Y}_1^{i+1:m},D^{i+1:m},F)\nonumber\\
  &~~~~~~+I(M^{i+1:m};\mathbf{Y}_1^{i+1:m},D^{i+1:m},F)\nonumber\\
  &~~~~~~~~-I(M^{i+1:m},E_{i+1},W;\mathbf{Y}_1^{i+1:m},D^{i+1:m},F)\nonumber\\
  &= I(M_i,E_i,W;\mathbf{Y}_1^{i:m},D^{i:m},F|M^{i+1:m})\nonumber\\
  &~~~~+I(M^{i+1:m};\mathbf{Y}_{1,i},D_i|\mathbf{Y}_1^{i+1:m},D^{i+1:m},F)\nonumber\\
  &~~~~~~-I(E_{i+1},W;\mathbf{Y}_1^{i+1:m},D^{i+1:m},F|M^{i+1:m}) \nonumber\\
  &= I(M_i,E_i,W;\mathbf{Y}_1^{i:m},D^{i:m},F|M^{i+1:m})\nonumber\\
  &~~~~-I(E_{i+1},W;\mathbf{Y}_1^{i+1:m},D^{i+1:m},F|M^{i+1:m}) \label{lemma3-1}\\
  &=I(M_i,E_i;\mathbf{Y}_1^{i:m},D^{i:m},F|M^{i+1:m})\nonumber\\
  &~~~~+I(W;\mathbf{Y}_1^{i:m},D^{i:m},F|M^{i:m},E_i)\nonumber\\
  &~~~~~~-I(E_{i+1},W;\mathbf{Y}_1^{i+1:m},D^{i+1:m},F|M^{i+1:m}) \nonumber\\
  &=I(M_i,E_i;\mathbf{Y}_{1,i},D_i,F|M^{i+1:m})\nonumber\\
  &~~~~+I(M_i,E_i;\mathbf{Y}_1^{i+1:m},D^{i+1:m}|M^{i+1:m},\mathbf{Y}_{1,i},D_i,F)\nonumber\\
  &~~~~~~+I(W;\mathbf{Y}_1^{i:m},D^{i:m},F|M^{i:m},E_i)\nonumber\\
  &~~~~~~~~-I(E_{i+1},W;\mathbf{Y}_1^{i+1:m},D^{i+1:m},F|M^{i+1:m}) \nonumber\\
  &=I(M_i,E_i;\mathbf{Y}_{1,i},D_i,F)+I(W;\mathbf{Y}_1^{i:m},D^{i:m},F|M^{i:m},E_i)\nonumber\\
  &~~~~-I(E_{i+1},W;\mathbf{Y}_1^{i+1:m},D^{i+1:m},F|M^{i+1:m}), \label{lemma3-2}
  \end{align}
  where (\ref{lemma3-1}) holds due to the independence between $M^{i+1:m}$ and $(\mathbf{Y}_{1,i},D_i)$, and (\ref{lemma3-2}) holds because $(M_i,E_i)$ and $(\mathbf{Y}_1^{i+1:m},D^{i+1:m})$ are independent, and $M^{i+1:m}$ is independent from both $(M_i,E_i)$ and $(\mathbf{Y}_{1,i},D_i,F)$, thus $I(M_i,E_i;\mathbf{Y}_{1,i},D_i,F|M^{i+1:m})=I(M_i,E_i;\mathbf{Y}_{1,i},D_i,F)$.  
  Then by Lemma \ref{lemma:1} and \ref{lemma:2} we have
  \begin{equation}
  L_i-L_{i+1}\leq O(N^3 2^{-N^{\beta}}).
  \end{equation}

\bibliographystyle{IEEEtran}
\bibliography{Polar_CICC}

\end{document}